%% file: main.tex
\documentclass{LMCS}

\def\doi{8 (1:05) 2012}
\lmcsheading%
{\doi}
{1--28}
{}
{}
{Mar.~\phantom09, 2011}
{Feb.~16, 2012}
{}

\usepackage{enumerate,hyperref}
\usepackage{pdfsync}
\usepackage{cite}
\usepackage{amsmath}
\usepackage{array}
\usepackage{mdwmath}
\usepackage{mdwtab}
\usepackage{eqparbox}
\usepackage{stfloats}

\usepackage{amssymb,mathrsfs}
\usepackage{graphicx}
\usepackage{version}

\includeversion{full}
\excludeversion{conf}

\begin{document}

\newcommand{\gwidth}[1]{|\!| #1 |\!|}
\newcommand{\dwidth}[2]{|\!| #1 |\!|_{#2}}
\newcommand{\self}{\mathtt{self}}
\newcommand{\child}{\mathtt{child}}
\newcommand{\descendant}{\mathtt{descendant}}
\newcommand{\parent}{\mathtt{parent}}
\newcommand{\nextsibling}{\mathtt{next}{-}\mathtt{sibling}}
\newcommand{\prevsibling}{\mathtt{prev}{-}\mathtt{sibling}}
\newcommand{\xeq}{\sim}
\newcommand{\xneq}{\not \sim}
\newcommand{\marg}[1]{}   
\newcommand{\word}{\mathrm{word}}
\newcommand{\gca}{\mathrm{gca}}
\newcommand{\aut}{{\mathcal A}}
\newcommand{\monoid}{{S}}
\newcommand{\triples}{\varphi}
\newcommand{\query}{\phi}
\newcommand{\size}[1]{{\mid} {#1} {\mid}}
\newcommand{\removed}[1]{}
\newcommand{\malyparagraph}[1]{\paragraph{\bf #1}}
\newcommand{\sredniparagraph}[1]{\paragraph{\bf #1}}

\newcommand{\set}[1]{\{#1\}}
\newcommand{\Aa}{{\mathcal A}}
\newcommand{\Bb}{{\mathcal B}}
\newcommand{\Nat}{{\mathbb N}}
\newcommand{\Int}{{\mathbb Z}}
\newcommand{\Oo}{{\mathcal O}}
\newcommand{\Cc}{{\mathcal C}}
\newcommand{\Ee}{{\mathcal E}}
\newcommand{\Gg}{{\mathcal G}}
\newcommand{\Tt}{{\mathcal T}}

\newtheorem{theorem}{Theorem}
\newtheorem{lemma}{Lemma}
\newtheorem{proposition}{Proposition}
\newtheorem{example}{Example}
\newtheorem{claim}{Claim}

\title{An extension of data automata that captures XPath}
\author[M.~Boja\'nczyk]{Miko{\l}aj Boja\'nczyk}
\address{Warsaw University}
\email{\{bojan,sl\}@mimuw.edu.pl}
\author[S.~Lasota]{S{\l}awomir Lasota} 
\address{\vskip-6 pt}

\thanks{This work has been partially supported by the Polish government grants
no. N206 008 32/0810 and N206 567840; and by the FET-Open grant agreement FOX, number
FP7-ICT-233599.}

\keywords{satisfiability of (regular) XPath queries, automata on data words and
trees, data automata}
\subjclass{F.1.1, H.2.3}




\maketitle

\begin{abstract}
We define a new kind of automata recognizing properties of data words or data trees and  
prove that the automata capture all queries definable in Regular XPath. 
We show that the automata-theoretic approach may be applied to
answer decidability and expressibility questions for XPath.
\end{abstract}


\input{intro}

\input{defs}

\input{automata}

\input{proof}

\input{simplifying}
\input{updown}
\input{applications}


\section*{Acknowledgment}
The authors would like to thank the anonymous reviewers 
for their valuable comments.

\bibliographystyle{plain}
\bibliography{IEEEabrv,citat}

\end{document}

%% file: intro.tex
\section{Introduction}

In this paper, we study data trees. In a data  tree, each node carries a label from a finite alphabet and a data value from an infinite domain. We study properties of  data trees, such as those defined in XPath, which refer to  data values only by testing if two nodes carry the same data value. Therefore we define a data tree as a pair $(t,\sim)$ where $t$ is a tree over a finite alphabet and $\sim$ is an equivalence relation on  nodes of~$t$. Data values are identified with equivalence classes of $\sim$.

Recent years have seen a lot of interest in automata for data trees and the special case of data words. The general theme is that it is difficult to design an automaton which  recognizes interesting properties and has decidable emptiness.

Decidable emptiness is important  in XML static analysis. A typical question of static analysis is the implication problem: given two  properties  $\varphi_1,\varphi_2$ of XML documents (modeled as data trees), decide if every document satisfying $\varphi_1$  must also satisfy $\varphi_2$. Solving the implication problem boils down to deciding emptiness of $\varphi_1 \land \neg \varphi_2$.

A common logic for expressing properties is XPath.
For XPath, satisfiability is undecidable in general, even for data words, see~\cite{DBLP:journals/jacm/BenediktFG08}. This means that most problems of static analysis are undecidable for XPath, e.g.~the implication problem.
Satisfiability is undecidable also for most other natural logics on data words or data trees, including first-order logic with predicates for order (or even just successor) and data equality.

The approach chosen in prior work was to find automata on data words or trees that would have decidable emptiness and recognize interesting, but necessarily weak, logics or fragments of XPath. 
These weak logics  include: fragments of XPath without recursion or negation~\cite{DBLP:journals/jacm/BenediktFG08,DBLP:conf/dbpl/GeertsF05}; first-order logic with two variables~\cite{BojanczykDMSS11, BojanczykJACM09}; forward-only fragments related to alternating automata~\cite{DBLP:journals/tocl/DemriL09,DBLP:conf/lics/JurdzinskiL07,DBLP:conf/pods/Figueira09,figueira-icdt2010}. The original automaton model for data words was~\cite{KaminskiF94}. See~\cite{DBLP:conf/csl/Segoufin06} for a survey.

In this paper, we take a different approach. Any model that  captures XPath will have undecidable emptiness. We are not discouraged by this, and try to capture XPath by something that feels like an ``automaton''. Three tangible goals are:  1. use the automaton to decide emptiness for interesting restrictions of data trees; 2.  use the automaton to prove easily that the automaton (and consequently XPath) \emph{cannot} express a property; 3.  unify other automata models that have been suggested for data trees and words.
 
What is our new model? To explain it, we use logic.
From a logical point of view, a nondeterministic automaton is a formula of the form $\exists X_1 \ldots \exists X_n\ \varphi(X_1,\ldots,X_n)$, where the kernel $\varphi$ is relatively simple, e.g.~it only talks about the relationship of labels in successive positions. As often in automata theory, when designing the automaton model,  we  try to use the prefix of existential set quantifiers as much as possible, in the interest of simplifying  the  kernel $\varphi$. For  satisfiability, this is like a free lunch, since deciding  satisfiability with or without the prefix are the same problem.

In the automaton model that we propose in this paper, the kernel $\varphi$ is of the form ``for every class $X$ of $\sim$, property $\psi(X,X_1,\ldots,X_n)$ holds'', where  $\psi$ is an MSO formula that can use  predicates for navigation (sibling order, descendant), predicates for testing labels from the finite alphabet, but not  the predicate $\sim$ for data equality. The data $\sim$ is only used in saying that $X$ is a class. In the case of data words, this model is an extension of the \emph{data automata} introduced in~\cite{BojanczykDMSS11}, which correspond to the special case when first-order quantifiers in  $\psi$ range  only over positions from $X$. For instance, our new model, but not data automata, can express the property ``between every two different positions in the same class there is at most one position outside the class with label $a$''.

The principal technical contribution of this paper is that the  model above can recognize all unary queries of XPath. This proof is difficult, and takes over ten pages. We believe the real value of this paper lies  in this proof, which demonstrates some powerful normalization techniques for formulas describing properties of data  trees. Since the scope of applicability for   these techniques will be clear only in the future; and since the appreciation of an ``automaton model'' may ultimately be a question of taste,  we  describe in more details the three tangible goals mentioned above.

1. The ultimate  goal of this research is to find  interesting classes of data trees which yield decidable emptiness for XPath. 
As a proof of concept, we define a simple subclass of data trees, called bipartite data trees,  and prove that emptiness of our automata (and consequently of XPath)  is decidable for bipartite data trees. This is only a preliminary result, we intend to find new subclasses in the future.

2.  We use the automaton  to prove that XPath cannot define certain properties. Proving inexpressibility results for XPath is difficult, because the truth value of an XPath query in a position $x$ might depend on the truth value of a subquery in a position $y < x$, which in turn might depend on the truth value of a subquery in a position $z > y$, and so on. On the other hand, our  automaton  works in one direction, so it is easier to understand its limitations. We use (an extension of) our automata to prove that for documents with two equivalence relations $\sim_1$ and $\sim_2$, some properties of two-variable first-order logic cannot be captured by XPath, which was an open question.

3. We use the automaton to classify existing models for data words in a single framework. A problem with the research on data words and data trees is that  the models are often incomparable in expressive power. In an upcoming paper\footnote{This paper is a journal version of a LICS 2010 paper~\cite{BL10}, which included a rough description of the classification mentioned in item 3. However, a thorough explanation of the classification requires much space, and uses different techniques than the results  in this paper. Therefore, we plan to present the classification in a separate paper.}, we will show that many existing models can be seen as syntactic fragments of our automaton. We hope that this classification will underline more clearly what the differences are between the models.


%% file: defs.tex
\section{Preliminaries}

\malyparagraph{Trees}
Trees are unranked, finite, and labeled by a finite alphabet $\Sigma$. We use the terms child, parent, sibling, descendant, ancestor, node in the usual way. The siblings are ordered. We write $x \leq y$ when  $x$ is an ancestor of $y$.
\begin{full}
Every nonempty set of nodes $x_1,\ldots,x_n$ in a tree has a greatest common ancestor (the
greatest lower bound wrt.~$\leq$), which is denoted $\gca(x_1,\ldots,x_n)$.
\end{full}

Let $t$ and $s$ be two trees, over alphabets $\Sigma$ and $\Gamma$,
respectively, that have the same sets of nodes. We write $t \otimes s$
for the tree over the product alphabet $\Sigma \times \Gamma$ that has the same
nodes as $s$ and $t$, and where every node has the label from $t$ on
the first coordinate, and the label from $s$ on the second coordinate.
If $X$ is a set of nodes in a tree $t$, we write $t \otimes X$ for the tree
$t \otimes s$, where $s$ is the tree over alphabet $\set{0,1}$, whose
nodes are the nodes of $t$ and whose labeling is the characteristic
function of $X$.

\malyparagraph{Regular tree languages and transducers} 
We use the standard notion of regular tree languages for unranked trees~\cite{CLT05}. 
We also use transductions, which map trees to trees.
Let $\Sigma$ be an input alphabet and $\Gamma$ an output alphabet.
A regular tree language $f$ over the product alphabet $\Sigma \times \Gamma$  can be interpreted as 
a binary relation, which contains pairs $
(s,t)$  such that $ s \otimes t \in f$.
We use the name \emph{letter-to-letter transducer} for such a relation, 
underlining  that  the trees in a pair $(s,t) \in f$ must have the same nodes. 
In short, we simply say transducer. Observe that the transducer is nondeterministic.
We often treat a transducer as a function that maps an input tree to a set of output trees, 
writing $t \in f(s)$ instead of $(s,t) \in f$.


\malyparagraph{Data trees}  A \emph{data tree} is a tree $t$ equipped with 
an equivalence relation $\sim$ on its nodes that represents data equality. 
We use the name \emph{class} for equivalence classes of $\sim$. 

\malyparagraph{Queries} Fix an input alphabet. We use the name  $n$-ary query for a function  $\query$ that   maps a
tree $t$ over the input alphabet to  a set $\query(t)$ of $n$-tuples of its nodes.  
In this paper, we  deal with queries of arities 0,1,2 and 3, which are called boolean, unary, binary and ternary. 
We also study queries that input a data tree $(t, \sim)$; they output a set of node tuples $\query(t, \sim)$ as well.

\malyparagraph{MSO} Logic is a convenient way of specifying queries, both for  trees and data trees.  We
use monadic second-order logic (MSO). In a given tree, or a data
tree, a formula of MSO is allowed to quantify over
nodes of the tree using individual variables $x,y,z$, and also over
sets of nodes using set variables $X,Y,Z$. A formula $\query$ with
free individual variables $x_1,\ldots,x_n$ defines an $n$-ary query,
which selects in a tree $t$ the set $\query(t)$ of tuples
$(x_1,\ldots,x_n)$ that make the formula true. To avoid confusion, we
use round parentheses for the tree input of a query, $\query(t)$, and
square parentheses for indicating the free variables of a query.  
The two parenthesis can appear
together, e.g.~$\query[x_1,\ldots,x_n](t)$ will be the set of
$n$-tuples selected in a tree $t$ by a query with free variables
$x_1,\ldots,x_n$.

When working over trees without data,  MSO formulas use binary predicates for the child and
next-sibling relations (that allow to define descendant and following-sibling relations), as well as a unary predicate for each label. Queries defined by MSO with these predicates
are called \emph{regular queries} (of course, regular queries can also be characterized in terms of automata). When working over data
trees, we  additionally allow a binary predicate $\sim$ to test
data equality. A query using $\sim$ is no longer called regular. For instance, the following unary query selects positions that have classes of size at least two:
\begin{align*}
\varphi(x) \qquad = \qquad \exists y \quad x \neq y \land x \sim y.
\end{align*}

\malyparagraph{Extended Regular XPath}
We define a variant of XPath that works over data trees.
For unary queries, the variant is an extension of XPath,  thanks to including MSO  as part of its syntax. 
We call the variant Extended Regular XPath. Unlike XPath, Extended Regular XPath allows for queries of
arbitrary arity.
Expressions of   Extended Regular XPath  are defined below. 
\begin{iteMize}{$\bullet$}
\item Let $\Gamma=\set{\query_1,\ldots,\query_n}$ be a set of already defined unary queries of Extended Regular XPath, which will be treated as unary predicates. In the induction base, the set $\Gamma$ is necessarily empty.
Suppose that $\varphi[x_1,\ldots,x_m]$ is  an MSO query that uses unary predicates for queries from $\Gamma$, unary predicates for letters of the input alphabet, and the binary child and next-sibling predicates.  Then $\varphi$  is an $m$-ary query of Extended Regular XPath. It is important that $\varphi$ does not introduce any new use of  the data equality predicate $\sim$, all appearances of $\sim$ are reserved to the queries from $\Gamma$.  

\item Suppose that $\varphi[x,y_1,y_2]$ is a ternary query of Extended Regular XPath.  Then the following property of $x$ is a unary query of Extended Regular XPath
\begin{eqnarray}
\label{e:unaryquery}
\exists y_1 \exists y_2 \quad y_1 \sim y_2 \land \varphi[x,y_1,y_2].
\end{eqnarray}
Likewise for $y_1 \not \sim y_2$ instead of $y_1 \sim y_2$.
\end{iteMize} 
The definition above allows for queries of any arity. In the paper, we will be principally interested in queries of arity one, and the queries of arity at most three  used to  build them. By abuse of nomenclature, we will write XPath instead of Extended Regular XPath.

\malyparagraph{Binary trees}
A binary tree is a tree where each node has 
at most two children.  Although the interest of XPath is mainly for unranked trees, we assume in the proofs that  trees are binary. This assumption can be made because XPath, as well as the models of automata introduced later on, are  stable under the usual first-child /
next-sibling  encoding in the following sense.
A language $L$ of unranked data trees can be expressed by a boolean XPath query if and only if the set of  binary encodings of trees from $L$  can be expressed by a boolean XPath query. A similar, though more technical, statement holds for unary queries.

Words will be considered as a special case of binary trees where each node has at most one child.


%% file: automata.tex
\section{Class automata}
\label{sec:class-automata}
In this section we define a new type of  automaton for data trees, called a {class
automaton},   and state the main result: class automata
capture all queries definable in XPath.

A \emph{class automaton} is a type of automaton that recognizes
properties of data trees.  A class automaton is given by: an
{input alphabet} $\Sigma$, a {work alphabet} $\Gamma$, a
nondeterministic letter-to-letter tree transducer $f$ from the input
alphabet $\Sigma$ to the work alphabet $\Gamma$, and a regular tree language on alphabet $\Gamma \times \set{0,1}$, called the {class condition}. The class automaton accepts a data tree
$(t,\sim)$ over input alphabet $\Sigma$ if there is some output $s \in
f(t)$ such that for every class $X$, the class condition contains the tree $s \otimes X$.

\smallskip
\begin{example}\label{example:1}
Consider an input alphabet $\Sigma= \set{a,b}$. Let $L$ be the  data trees 
 where
some class contains at least three nodes with label $a$. This language is recognized by a class automaton. The work alphabet is $\Gamma=\set{a,c}$. The transducer 
guesses three nodes with label $a$, and outputs $a$ on them, other nodes get $c$. The class condition consists of trees $s \otimes X$ over alphabet $\Gamma \times \set{0,1}$ where $X$ contains all or none of the nodes with label $a$. Note that the class condition does not  inspect positions outside $X$. 
\end{example}

\smallskip
\begin{example}\label{example:2}
Let $K$ be the set of data words over $\Sigma = \set{a,b}$
where each class has exactly two positions $x< y$, and there is at most one $a$ in the positions $\set{x+1,\ldots,y-1}$.
In the class automaton recognizing $K$, the transducer is the identity function, and the class  condition is
\begin{eqnarray*}
	\Sigma_0^* \cdot \Sigma_1 \cdot  b_0^*  \cdot (a_0+\epsilon)  \cdot b_0^* \cdot \Sigma_1 \cdot \Sigma_0^*
\end{eqnarray*}
where $\Sigma_i$ is a shortcut for $\Sigma \times \set{i}$, likewise for $a_i$ and $b_i$.
\end{example}

\sredniparagraph{Comparison to data automata}
Class automata are closely related to \emph{data automata} introduced in~\cite{BojanczykDMSS11}. Data automata were defined for data words. Since  it is not clear what the correct tree version thereof is, we just present the version for data words. Like a class automaton,  a data automaton has an input alphabet $\Sigma$, a work alphabet $\Gamma$, and a nondeterministic letter-to-letter transducer $f$ (this time only for words). The difference is in the class condition, which is less powerful in a data automaton. In a data automaton, the class condition is a word language over $\Gamma$, and not $\Gamma \times \set{0,1}$. The data automaton accepts a data word $(w,\sim)$ if there is some output $v \in f(w)$ such that for every class $X$, the class condition contains the subsequence of $v$ obtained by only keeping positions from $X$.  In the realm of data words, data automata can be seen as a special case of class automata, where the class condition is only allowed to look at positions from the current class. The language $L$ in Example~\ref{example:1} can be recognized by a data automaton (in the case of words), while the language $K$ in Example~\ref{example:2} is a language that can be recognized by class automata, but not data automata. 

The difference between data automata and class automata is crucial for decidability of emptiness.
Data automata have decidable emptiness~\cite{BojanczykDMSS11}, the proof being a reduction to reachability in Vector Addition Systems with States.  Class automata have undecidable emptiness, because they capture the logic XPath, which has undecidable satisfiability. Also, a direct and simple proof of undecidable emptiness for class automata can be given, by encoding runs of two-counter machines, without going through the difficult reduction from XPath.

\removed{
It is not difficult to show that languages recognized by class automata are closed under projections, union and intersection. (The projection of a language of data words over an alphabet $\Gamma$ is obtained by applying a function $\pi : \Gamma \to \Sigma$ to each label in each data tree.)}

\sredniparagraph{Closure properties} 
Suppose that $f : \Sigma_1 \to \Sigma_2$ is any function. We extend $f$ to a function $\hat f$ from data trees over alphabet $\Sigma_1$ to data trees over alphabet $\Sigma_2$, by just changing the labels of nodes, and not the tree structure or data values. We use the name relabeling for any such function $\hat f$.
\begin{lemma}\label{lem:closure}
	Languages of data trees recognized by class automata are closed under union, intersection,  images under relabelings, and inverse images under relabelings.
\end{lemma}

\begin{proof}
The inverse images are the simplest: the letter-to-letter tree transducer in the class automaton is composed with the relabeling.
For intersection, one uses Cartesian product.
For union and images under relabelings, one uses nondeterminism. 
\end{proof}

\sredniparagraph{Evaluation}
The evaluation  problem (given an automaton and a data word/tree,
check if the latter is accepted by the former) is NP-complete, even for a fixed data automaton
(cf.~\cite{BjorklundS07}). Hence it is also NP-complete for class automata, which extend data automata.
\removed{
Hardness is easily shown by reduction from the 3-colorability problem:
given a graph $G$, encode its edges by separate data values, 
and vertices as blocks of consecutive positions in a data word.
The transducer of the automaton $\Aa_G$ guesses a coloring and the class condition checks
color inequality along every edge.
$G$ is 3-colorable if and and only if $\Aa_G$ accepts the encoding of $G$.
}

\sredniparagraph{Class automata as a fragment of MSO} As mentioned in the introduction, one can see a
class automaton as a restricted type of formula of monadic
second-order  logic. This is a  formula of the form:
\begin{align}
  \label{eq:mso-form}
  \exists X_1 \cdots \exists X_n \  \forall X  \  \mathrm{class}(X) \Rightarrow
 \varphi(X_1,\ldots,X_n,X)
\end{align}
where $X_1,\ldots,X_n,X$ are variables for sets of nodes, the $\mathrm{class}$ formula is 
\begin{equation*}
	\mathrm{class}(X) = \qquad \exists y \forall x  \quad x \in X \iff y \sim x
\end{equation*}
and $\varphi$ is a formula of MSO that does not use $\sim$. 
Formulas of the above form recognize
exactly the same languages of data trees as class automata. For translating a class automaton to a formula, one uses the variables $X_1,\ldots,X_n$ to encode the output of the transducer, and the formula $\varphi$ to test two things: a) the variables $X_1,\ldots,X_n$ encode a legitimate output of the transducer; and b) the class condition holds for $X$.

\sredniparagraph{Main result}
The main result of this paper is Theorem~\ref{thm:main} below, which says that unary XPath queries over data trees can be recognized by class automata. To state the theorem, we need to say how a class automaton  recognizes a unary query. We do this by encoding a unary query $\query$ over data trees as a language of data trees:
\begin{align*}
L_\query =  \set{ (t \otimes X, \sim): \mbox{$(t,\sim)$ is a data tree, $X = \query(t,\sim)$}} .
\end{align*}
In other words, the language consists of data trees decorated with the set of nodes selected by the query. This encoding does not generalize to binary queries.

\begin{theorem}\label{thm:main}
  Every unary XPath query over data trees can be recognized by a class automaton.
\end{theorem}


%





We begin the proof of Theorem~\ref{thm:main}, mainly to show where the difficulties appear. Then, we lay out the proof strategy in more detail.
When referring to the language of a unary query, we mean the  encoding above.

We do an induction on the size of the unary query.
The base case, when the query is a label $a$, is straightforward. 
Consider now the induction step, with a unary query 
\begin{eqnarray*}
\query[x] = \quad
\exists y_1 \exists y_2 \quad y_1 \sim y_2 \land \varphi[x,y_1,y_2]
\end{eqnarray*}
as in~\eqref{e:unaryquery}. (The same argument works for the case where $y_1 \not \sim y_2$.)
 Let  $\query_1,\ldots,\query_n$ be all the unary XPath subqueries that appear in $\varphi$. By the induction assumption, the languages of the subqueries are recognized by class automata $\Aa_1,\ldots,\Aa_n$. Let the variables $X,X_1,\ldots,X_n$ denote sets of nodes. Consider the set $L$ of data trees
\begin{align*}
	(t \otimes  X  \otimes X_1 \otimes \cdots \otimes X_n, \sim) 
\end{align*}
such that a) for each $i \in \set{1,\ldots,n}$, the data tree $(t \otimes X_i, \sim)$ is accepted by the automaton $\Aa_i$; and b) $X$ is the set of nodes selected by the query $\query'$ obtained from $\query$ by replacing each subquery $\query_i$ with ``has $1$ on coordinate corresponding to $X_i$''. Suppose that the language of $\query'$ is recognized by a class automaton. Then so is $L$, by closure of class automata under intersection and  inverse images of projections, see Lemma~\ref{lem:closure}. Finally, the language of $\query$ is the image of $L$ under the projection which removes the labels describing the sets $X_1, \ldots, X_n$. 

It remains to show that $\query'$ is recognized by a class automaton (the advantage of $\query'$ over $\query$ is that it uses data equality $\sim$ only once, to say that $y_1\sim y_2$).   A major part of this paper is devoted to this case, which is stated in the following proposition.

\begin{proposition}\label{prop:goal}
	Class automata can recognize queries 
	\[
		\query[x] = \quad
	\exists y_1 \exists y_2 \quad y_1 \sim y_2 \land \varphi[x,y_1,y_2],
	\]
	 where $\varphi$ is a regular ternary query (i.e.~$\varphi$ does not use $\sim$). Likewise for $y_1 \not \sim y_2$.
\end{proposition}

\sredniparagraph{Proof strategy}
The construction of the automaton for $\query[x]$ is spread across several sections. In
Section~\ref{subsec:witness-functions-bounded-width}, we introduce the main 
concepts underlying the proof.
In particular, we define a new complexity measure for binary relations on tree domains, 
called guidance width,   
that seems to be of independent interest.
In Section~\ref{subsec:startproof}
we start the proof itself, formulate
an induction, and reduce Proposition~\ref{prop:goal} to 
a more technical Theorem~\ref{thm:bounded-width}.
Then in Section~\ref{sec:simplifying-delta} we
identify a simplified form of queries appearing in Theorem~\ref{thm:bounded-width}
and show how arbitrary queries can be transformed to the simplified form.
Finally, Section~\ref{sec:diff-arrang-trees} 
contains the proof of Theorem~\ref{thm:bounded-width} for these simplified queries,
the heart of the whole proof. 


%% file: proof.tex
\subsection{Discussion of the proof}

In this section, we discuss informally the concepts that appear in the proof of Theorem~\ref{thm:main}.
For the purpose of illustration, we use words.

We begin our discussion with words without data. For  a regular binary query $\varphi[x,y]$, consider  the unary query 
\begin{align*}
\psi[x] =	 \exists y \; \varphi[x,y].
\end{align*}
We use the name witness function in a word $w$  for a function which maps every position $x$ satisfying $\psi$ to some $y$ such that $\varphi[x,y]$ holds. Consider, as an example, the case where $\varphi[x,y]$ says that there exists exactly one $z$ that has label $a$ and satisfies $x < z < y$. The following picture shows a witness function.
\begin{center}
\includegraphics[scale=0.7]{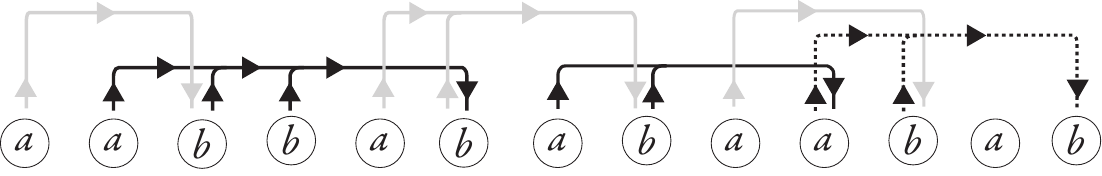}
\end{center}

The way the picture is drawn is important. The witness function is recovered by following arrowed lines.
The arrowed lines are colored black, dashed black, or gray, in such a way that no position is traversed by two  arrowed lines of the same color. With the formula $\psi$ in the example, any input word has a witness function that can be drawn with three colors of arrowed lines. This can be generalized to arbitrary MSO binary queries; the number of  colors depends only on the query, and not the input word.

The above observation may be used to design a nondeterministic automaton recognizing a property like $\forall x \; \psi[x]$. The automaton would guess the labeling by arrows and then verify its correctness. The number of states in the automaton would grow with the number of colors; hence  the need for a bound on the number of colors.  Of course, there are other ways of recognizing $\forall x\ \psi[x]$, but we talk about the coloring since this is the technique that will work with data.

We now move to data words. Consider a unary query
\begin{align*}
	\begin{array}{ll}
		\psi[x] = &\exists y_1 \exists y_2\  \quad y_1 \sim y_2 \ \land \ \varphi[x,y_1,y_2] ,
	\end{array}
\end{align*} where  $\varphi[x,y_1,y_2]$  says that $y_1 < x \le y_2$, there is exactly one $a$ label in the positions $\set{y_1,\ldots,x-1}$ and there is exactly one $b$ label in the positions $\set{x,\ldots,y_2}$. The query $\psi[x]$ is an example of a query as in Proposition~\ref{prop:goal}. Consider  the following data word (the labels are blank, $a$ and $b$, the data values are 1 \ldots 6).\\
\begin{center}
	\includegraphics[scale=0.7]{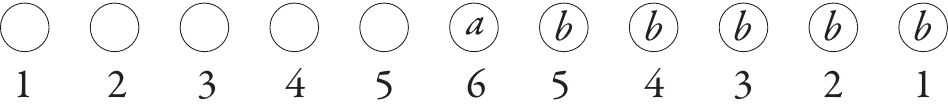}
\end{center}
Say $x$ is  the first node with label $b$. This node is selected by~$\psi$. Consider the pairs $(y_1, y_2)$
required by $\psi[x]$, which we call \emph{witnesses}. 
The only possibility for $y_2$  is $x$ itself;
thus $y_1$ is also determined, as the only other position with the same data value.  So there is only one witness pair.
The same situation holds for all other positions with label $b$, which  are the only positions selected by $\psi$.  The drawing below shows how witness pairs are assigned to positions.
\begin{center}
	\includegraphics[scale=0.7]{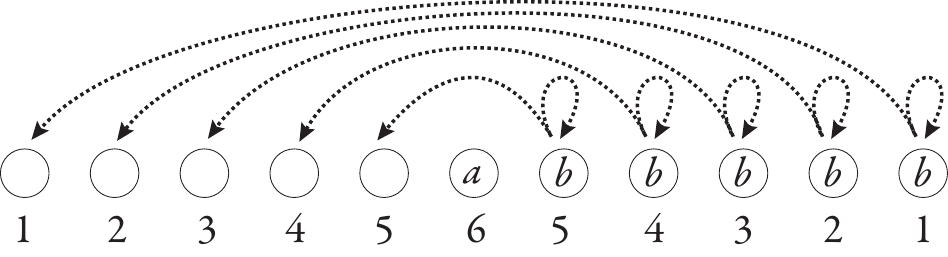}
\end{center}

We would like to draw this picture with  colored arrows, as we did for the first example of witness functions. If we insist on drawing arrows that connect each position $x$ with its corresponding witness $y_1$, then we will need 5 colors
as the middle position (labeled by a) is traversed by 5 arrows; 
the picture also generalizes to any number of colors. On the other hand, connecting each position $x$ with its corresponding $y_2$ (a self-loop) requires only one color. We  can symmetrically come up with instances of data words where connecting each node $x$ to $y_2$ requires 
an unbounded number of colors.

A consequence of our main technical result, Theorem~\ref{thm:bounded-width}, is that  a bounded number of colors  is sufficient if we want to perform the following task: for each position $x$ selected by $\psi$, choose some witness pair $y_1 \sim y_2$, and connect $x$ to  \emph{either} $y_1$ \emph{or} $y_2$.
The bound depends only on $\psi$; in particular, the bound does not depend on input data tree.

The concepts of witness functions and coloring are defined more precisely below.

\subsection{The core result}
\label{subsec:witness-functions-bounded-width}

We will state some technical results for a  structure more general than a data tree, namely a graph tree. A \emph{graph tree} is a tree $t$ endowed with an arbitrary symmetric binary relation $E$ over its nodes.
A data tree is the special case of a graph tree where $E$ is an equivalence relation.

\malyparagraph{Witness functions}
Let $\triples[x, y_1, y_2]$ be a regular  query (think of Proposition~\ref{prop:goal}), and consider a graph tree $(t,E)$.
We are interested in triples $(x,y_1,y_2)$ selected by $\triples$ 
in $t$ such that $(y_1, y_2) \in E$.  (Think of $E$ being either the data equivalence relation $\sim$, or its complement.)
Consider any such triple. The node $x$
is called the \emph{source node};
the notion of source node is  relative to the query $\triples$
and relation $E$,
which will usually be clear from the context and  not mentioned
explicitly.
The pair $(y_1,y_2)$ is called the \emph{witness pair},
$y_1$ is called the \emph{first witness}, and $y_2$ is called the
\emph{second witness}. These notions 
are all relative to a given $x$, but if we do not
mention the $x$, then $x$ is quantified
existentially. 
Let $X$ be a set of nodes in a graph tree (not necessarily containing all
nodes). A \emph{witness function} for $\triples$ and $X$ in a graph tree is
a  function which maps every node $x \in X$, treated as a source node, to some (first
or second) witness. There may be many witness functions, since for
each node we can choose to use either a first witness or a second
witness, and there may be multiple witness pairs.

The key technical result of this paper is that one can always find a   witness function of low complexity. 
The notion of complexity is introduced below.

\sredniparagraph{Guidance width}
A \emph{guide} in a tree $t$ is given by two nonempty sets of \emph{source nodes} and 
\emph{target nodes}. The \emph{support} of the guide is the set of all
nodes and edges on (the shortest) paths that connect some source node with a target node,
including all the source and target nodes. 
A guide \emph{conflicts} with another guide if their supports
intersect. We  write $\pi$ for guides.


A guidance system is a set of guides $\Pi$. 
It induces a relation containing all pairs $(x, y)$ of tree nodes such that
$x$ is a source and $y$ a target in some guide in $\Pi$.
An $n$-color guidance
system is a guidance system whose guides can be colored by $n$ colors
so that conflicting guides have different colors.
The \emph{guidance width} of a binary relation $R$ on tree nodes
is the smallest $n$ such that some $n$-color guidance system induces $R$.

In the proof we will only consider guidance systems for
relations $R$ that are partial functions from tree nodes to tree nodes.
In such cases, it is sufficient to restrict to \emph{deterministic} guides,
i.e., those with precisely one target node.
From now on, if not stated otherwise, a guidance system will be implicitly
assumed to contain only deterministic guides.

\sredniparagraph{Witness functions of bounded width}
\label{subsec:wf-of-bw}
We are now ready to state the main technical result, which forms the core of the proof of Theorem~\ref{thm:main}.
\begin{theorem}\label{thm:bounded-width}
  Let $\triples$ be a regular ternary query.
  There exists a constant $m$, depending only on $\triples$, such that
  in every graph tree, every set of source nodes has some witness
  function of guidance width at most $m$.
\end{theorem}

In other words,  regular ternary queries have
witness functions of \emph{bounded guidance width}.
 Before
proving the theorem,
we show  how it implies Theorem~\ref{thm:main}.

\subsection{From Theorem~\ref{thm:bounded-width} to Proposition~\ref{prop:goal}}  
\label{subsec:startproof}
We  show how Theorem~\ref{thm:bounded-width} implies the last remaining piece of Theorem~\ref{thm:main}, 
namely Proposition~\ref{prop:goal}.  Consider a unary query $\query[x]$ as in the statement of Proposition~\ref{prop:goal}. 
We begin with the case when $\query[x]$  is of the form
\begin{align*}
	\begin{array}{ll}
		\query[x] = &\exists y_1 \exists y_2\  \quad y_1 \sim y_2 \ \land \ \varphi[x,y_1,y_2].
	\end{array}
\end{align*}
We need to find a class automaton that accepts the
data trees $(t \otimes X, \sim)$ where $X$ is the set of all nodes 
selected by $\query$ in the data tree $(t, \sim)$.
The class automaton will test the conjunction of two properties:

\medskip

\noindent{\it Completeness.} Each node selected by $\query$  in $(t,\sim)$ is in $X$. \\
\noindent {\it Correctness.} Each node in $X$ is selected by $\query$ in $(t,\sim)$.  \medskip

Recall that $\varphi$ is a regular query.
We give separate class automata for the two
properties. Completeness is  simple. It can be rephrased as
\begin{quote}
  for every class $Y$ and  triple $(x, y_1, y_2)$ selected by $\varphi$,  if $y_1, y_2 \in Y$ then $x \in X$.
\end{quote}
This is the type of property class automata are designed for: for every class, test a regular property. (Recall the discussion on class automata as a fragment of MSO.) Correctness is the difficult property, since the order of quantifiers is not the same as in a class automaton:
\begin{quote}
  for every  $x \in X$ there is a class $Y$ and   $y_1,y_2 \in Y$  such that  $(x,y_1,y_2)$ is selected by $\varphi$.
\end{quote}
Our solution is to use, as a part of the class automaton to be designed,
a guidance system
given by Theorem~\ref{thm:bounded-width}. 
%
%



Apply Theorem~~\ref{thm:bounded-width} to  $\triples$, yielding a constant $m$.
The class automaton for the correctness property works as follows. 
Given an input data tree $(t \otimes X,\sim)$, it guesses an $m$-color guidance system;
let $R$ stand for the induced relation.
The automaton then checks the two conditions below.
\begin{iteMize}{A.}
\item For every $x \in X$ there is some $y$ with $x R y$.    
\item For every class $Y$,  if $x R y$,
  $x \in X$, $y \in Y$, then 
  either $(x,y,y')$ or $(x, y', y)$ is in $\triples(t)$, 
  for some $y' \in Y$.
\end{iteMize}
If the class automaton accepts,
then clearly every position in $X$ is a source node. Conversely, if all nodes
in $X$ are source nodes, then there is an accepting run of the above
class automaton. This accepting run uses the guidance system
for the witness function from Theorem~\ref{thm:bounded-width}.

This completes the proof for the case when $\query[x]$ requires $y_1 \sim y_2$. For the case $y_1 \not \sim y_2$, the proof is almost the same, 
except for two changes. 
The first change is that we apply Theorem~\ref{thm:bounded-width} to 
the graph trees $(t,E)$, obtained from data trees $(t, \sim)$ by taking as $E$ the complement of $\sim$. 
This explains why  Theorem~\ref{thm:bounded-width} is formulated for 
graph trees and not just data trees. 
The  second change is that we write $y' \not \in Y$ instead of $y' \in Y$ at the end of  condition B.

%% file: simplifying.tex
\section{Simplifying the query}   
\label{sec:simplifying-delta}
\newcommand{\join}{\otimes}


Before proving Theorem~\ref{thm:bounded-width}, we formulate two simplifying conditions about the 
regular query  $\triples[x,y_1,y_2]$. 

For two nodes $x, y$ in a tree $t$, we write $\word_t(x, y)$ for the sequence of labels on the
unique shortest path from $x$ to $y$ in $t$, including $x$ and $y$.  We  omit the
subscript $t$ when a tree is clear from the context.  Note that
$\word_t(x, y)$ is always nonempty and $\word_t(x,x)$ is the label of $x$.

The two  conditions about the query  $\triples[x,y_1,y_2]$ are:
\begin{enumerate}
\item  All  selected triples satisfy $y_1 < x < y_2$.
\label{en:simplify1}
\item Whether or not a triple is selected depends only on  the words
	$\word_t(y_1,x)$ and $\word_t(x,y_2)$. It does not depend on nodes outside the path from $y_1$ to $y_2$.
\label{en:simplify2}
\end{enumerate}

The goal of this section is to reduce
Theorem~\ref{thm:bounded-width} to the case when $\triples[x,y_1,y_2]$
is a \emph{simplified query}, as defined above.
This simplification is achieved in several steps.
(In the case of words, the simplification would be standard, but for trees it  requires new ideas about guidance systems.)
Formally, in this section we show that  Theorem~\ref{thm:bounded-width} follows
from Theorem~\ref{thm:simplified-bounded-width} (deliberately formulated as late as in the forthcoming 
Subsection~\ref{sec:arrangm-eqref-eqref}) that only speaks about simplified queries.
Theorem~\ref{thm:simplified-bounded-width} itself is proved in Section~\ref{sec:diff-arrang-trees}.

\subsection{Generalized witness functions}

Fix any number $n \in \Nat$, although we will be mainly interested in
$n \in \set{1,2}$. Consider a regular query $\triples[x,y_1,\ldots,y_n]$
over trees. 
Consider now a tree $t$ together with a set
$E$ of $n$-tuples of nodes in $t$. As before, the idea is that $E$ gives a
constraint on the witness variables.  A \emph{witness tuple} for a node $x$
is a tuple $(y_1,\ldots,y_n) \in E$ such that $(x,y_1,\ldots,y_n)$ is
selected by $\triples$ in $t$. In this case, we say that 
$x$ is a \emph{source}, and $y_i$ is an \emph{$i$-th witness} for $x$ 
(the other variables 
are quantified existentially).

A \emph{witness function} for $\triples$ and a set of source nodes $X$ in $(t,E)$ is a
function which assigns to each node $x \in X$ some witness (an $i$-th
witness for some $i$, with $i$ depending on $x$). 

We say that a regular query $\triples[x,y_1,\ldots,y_n]$
has \emph{witness functions of guidance width $m$}
if for every tree $t$, every choice $E$ of
$n$-tuples of nodes of $t$ and every set $X$ of source nodes in $(t,E)$, there is a witness function for $\triples$ 
and $X$ of guidance width at most $m$.
A query $\triples$ has witness functions of \emph{bounded guidance width} if some
such $m$ exists.

\subsection{Three arrangements}
\label{sec:three-arrangements}
By an \emph{arrangement} of the nodes $x,y_1,y_2$ in a tree we mean 
the information on how these nodes, and their greatest common ancestors
\begin{align*}
  \gca(x,y_1) \quad   \gca(x,y_2) \quad  \gca(y_1,y_2) 
\end{align*}
are related with respect to the descendant ordering. We distinguish
three different arrangements, pictured below.
\begin{center}
  \includegraphics[scale=0.6]{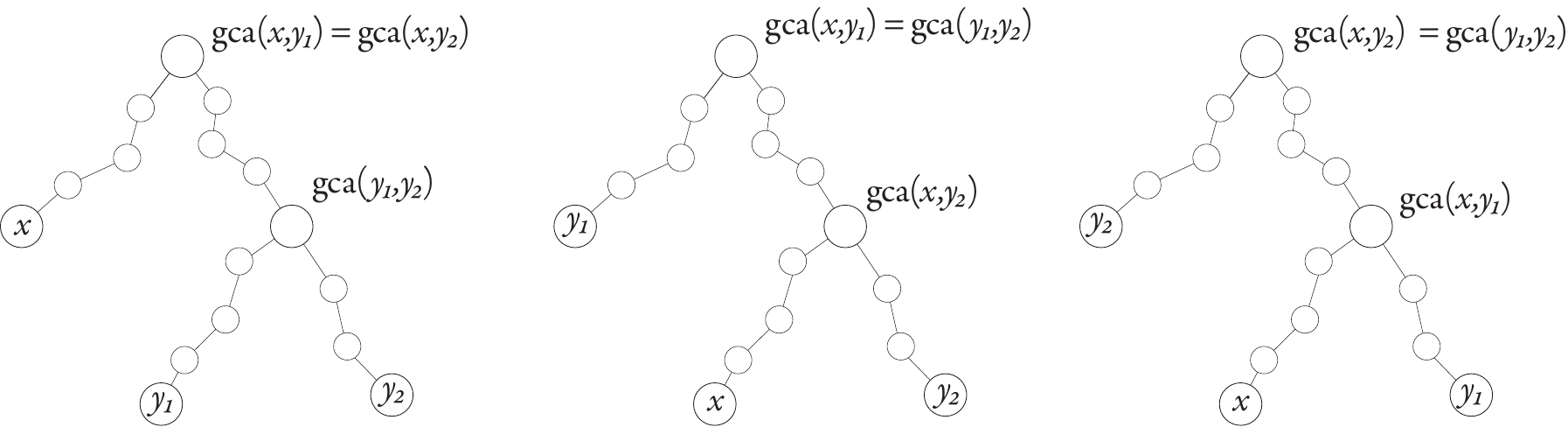}
\end{center}
These arrangements correspond, respectively, to the following situations.
\begin{align}
    \gca(x,y_1) =   \gca(x,y_2) \le  \gca(y_1,y_2) \label{eq:a1} \tag{A1}\\
    \gca(x,y_1) =   \gca(y_1,y_2) \le  \gca(x,y_2) \label{eq:a2} \tag{A2}\\
    \gca(x,y_2) =   \gca(y_1,y_2) \le  \gca(x,y_1) \label{eq:a3} \tag{A3}
\end{align}
The arrangements are not contradictory, for instance the case
$x=y_1=y_2$ is covered by all three.  
The slightly more general case $y_1 = y_2$ that essentially represents binary queries $\triples[x,y]$,
is fully covered by~\eqref{eq:a1}.

\begin{lemma}
We may assume without loss of generality that all the triples selected
by $\triples$, as in the statement of Theorem~\ref{thm:bounded-width}, have the same arrangement.  
\end{lemma}
\begin{proof}
   Otherwise we can split $\triples$ into a union of three
queries, one for each arrangement, and then combine the three separate
guidance systems.
\end{proof}

\subsection{Path-based queries}
\label{sec:path-based-queries}

 Let us fix one of the
arrangements. 
There are four words $w_1,w_2,w_3,w_4$ that will interest us. These
are shown on the picture below  for the arrangement~\eqref{eq:a1} only, but
the reader can easily see the situation for all other arrangements.
\begin{center}
  \includegraphics[scale=0.6]{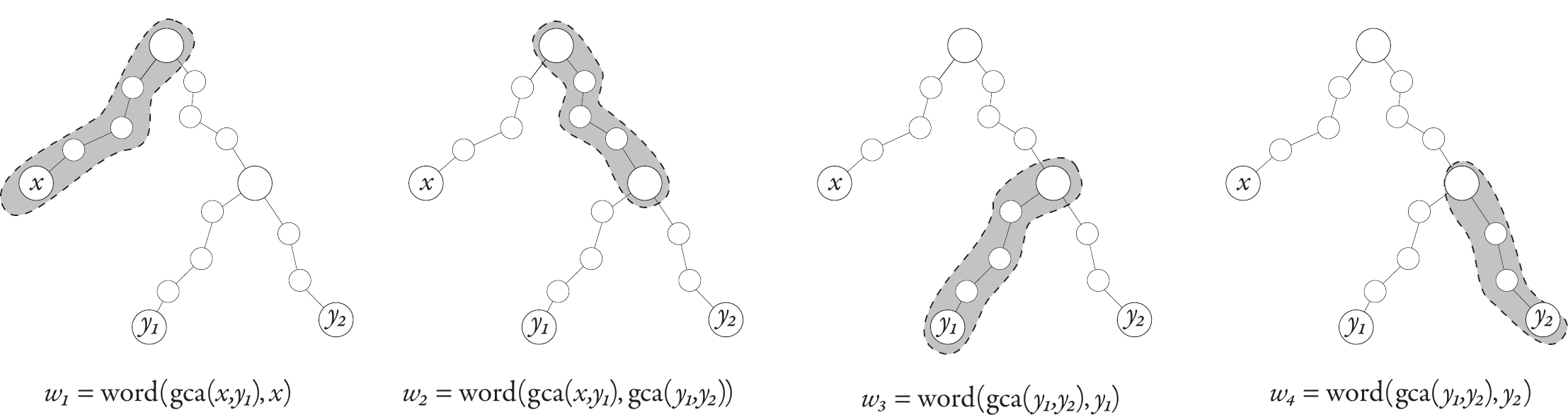}\ .
\end{center}
A regular query $\varphi[x,y_1,y_2]$ is called \emph{path-based} if its truth value depends only
on some regular properties of the four words $w_1,\ldots,w_4$. The precise definition of path-based queries we use in this paper is in terms of monoids. A query that selects triples only in arrangement~\eqref{eq:a1} is called \emph{path-based} if there exists a monoid morphism
\begin{eqnarray*}
  \alpha : \Sigma^* \to S
\end{eqnarray*}
such that membership $(x,y_1,y_2) \in \triples(t)$ depends only on
the values assigned by $\alpha$ to the words $w_1, \ldots, w_4$.
In other words, there is a set of \emph{accepting quadruples}
$F \subseteq S^4$ such that  $(x,y_1,y_2)$ belongs to $\triples(t)$ if and 
only if 
\begin{align*}
  (\alpha(w_1), \ldots, \alpha(w_4)) \in F\ .
\end{align*}
An analogous definition of path-based queries is given for the other arrangements~\eqref{eq:a2} and~\eqref{eq:a3}.

\begin{lemma}
\label{lemma:path-based}
  We may assume without loss of generality that $\triples$, as in the statement of 
  Theorem~\ref{thm:bounded-width},   is path-based.
\end{lemma}
\begin{proof}
The key observation is that   there is a functional letter-to-letter transducer
  $f$ and a path-based query $\gamma$ such that $\triples = \gamma
  \circ f$, i.e.~for a tree $t$ the set $\triples(t)$ of tuples selected
  by $\triples$ in $t$ is the same as the set of tuples selected by
  $\gamma$ in $f(t)$. The observation can be proved using logical methods
  (the transducer computes MSO theories) or using automata methods
  (the transducer computes state transformations).

To prove the lemma, we need to show that if
Theorem~\ref{thm:bounded-width} is true for the path-based queries
$\gamma$, then it is also true for arbitrary ternary queries $\triples$. But this is
straightforward, as $\triples$ and $\gamma$ have the same witness functions
in trees $t$ and $f(t)$, respectively.
\end{proof}

\subsection{Composing guidance systems}

In the sequel we will  compose guidance systems as outlined in the lemma 
below.
For two partial functions $f, g$ on the set of nodes of a tree,
by the composition $g \circ f$ we mean, somewhat non-standardly, 
the partial function with the same domain as $f$, and defined as follows:
\[
(g \circ f) (x) =  
\begin{cases}
 g(f(x)) & \text{ if g is defined on } f(x) \\
f(x) & \text{ otherwise.}
\end{cases}
\]

\begin{lemma}
\label{lemma:comp-guid-systems}
Let $f, g$ be partial functions on the set of nodes of a tree,
of guidance width $m_1$ and $m_2$, respectively.
Then their composition $g \circ f$ is of guidance width at most $2 m_1 m_2$.
\end{lemma}
\begin{proof}
Fix a tree $t$ together with some
$m_1$- and $m_2$-color guidance systems $\Pi_f$ and $\Pi_g$, inducing $f$ and $g$, respectively.
We will show existence of a $2 m_1 m_2$-color guidance system for $g \circ f$.

As the first step, combine $\Pi_f$ and $\Pi_g$ as follows:
a node $x$ is first guided by $\Pi_f$, and then, if $g$ is defined on $f(x)$, guided by $\Pi_g$
to its final destination.
Formally, $\Pi$ contains those guides of $\Pi_f$ whose destination node is not in the domain of $g$;
and moreover a number of guides that are composed of at least two guides, to be described now.

Fix a pair of colors $(k, l)$, where $k$ is a color used in $\Pi_f$ and $l$ is a color used in $\Pi_g$.
A composed guide, colored by the pair $(k, l)$, is derived from one $l$-colored guide from $\Pi_g$,
say $\pi$, and all those $k$-colored guides from $\Pi_f$ whose destination node is a source node of $\pi$,
The source nodes of the composed guide are all source nodes of all the above-mentioned $k$-colored
guides from $\Pi_f$. The target node of the composed guide is the target node of $\pi$.

We will focus on the composed guides only. 
(A 'non-composed' guide in $\Pi$, say colored $k$, may be safely considered 
as colored by $(k, l)$, for any $l$.)

The above coloring, using $m_1 m_2$ colors, is not satisfactory as same colored guides may be in conflict.
We will show how to resolve these conflicts by introducing an additionally distinguishing
piece of data into the colors.
Fix a color pair $(k, l)$ as above.
Note that a  conflict may only arise when the $\Pi_g$-part ($l$-colored in $\Pi_g$)
of one $(k, l)$-colored guide,
say $\pi_1$, conflicts with the $\Pi_f$-part ($k$-colored in $\Pi_f$) 
of another same colored guide, say $\pi_2$.
Consider an undirected graph $G$, whose nodes are all $(k, l)$-colored guides; there is an edge
between $\pi_1$ and $\pi_2$ in the graph if the abovementioned conflict arises.

We claim that the graph $G$ is a forest, i.e., a disjoint union of trees.
Towards a contradiction, suppose that $G$ has a cycle consisting of $n$ pairwise different guides
$\pi_1, \ldots, \pi_n$. Take $\pi_{n+1} = \pi_1$.
Let $x_1, \ldots, x_n$ denote arbitrarily chosen nodes witnessing the conflicts,
i.e., $x_i$ belongs to the supports of guides $\pi_i$ and $\pi_{i+1}$.
In $\pi_{i+1}$, for any $i \leq n$, there is a unique path from $x_i$ to $x_{i+1}$ (take $x_{n+1}$ as $x_1$),
denote it $p_i$;
$p_i$ always uses a path of a guide from $\Pi_f$, colored $k$, and a path of a guide from $\Pi_g$,
colored $l$.
As the $k$-colored guides never conflict, and likewise the $l$-colored ones,
the $k$-colored part of $p_i$
is separated from the same colored part of $p_{i+1}$ by at least one $l$-colored edge; 
thus the paths $p_i$ are nonempty, i.e., $x_i \neq x_{i+1}$.
Assume that $x_1, \ldots, x_n$ are pairwise distinct (if this is not the case, i.e.,
$x_i = x_j$, consider $x_i, \ldots, x_{j-1}$ instead; and consider $\pi_i, \ldots, \pi_{j-1}$ instead of
$\pi_1, \ldots, \pi_n$).

Now we are prepared to obtain a contradiction, thus proving that $G$ is a forest.
If two paths $p_i$ and $p_{i+1}$ share an edge adjacent to $x_{i+1}$, 
the edge may be removed from both paths;
this clearly forces $x_{i+1}$ to be replaced appropriately.
Thus the paths can be made edge-disjoint; moreover we keep the $x_i$ nodes pairwise distinct,
argued as above. 
Hence the paths $p_1, \ldots, p_n$ form a cycle in the tree $t$, a contradiction.

Knowing that $G$ is a forest, we may easily label its nodes by two numbers ${1, 2}$, 
level by level, starting from an arbitrary leaf in any connected component.
This additional numbering, added to the colors of the guides in $\Pi$, eliminates the problematic
conflicts and makes $\Pi$ a $2 m_1 m_2$-color guidance system as required.
\end{proof}

\subsection{Binary queries}
\begin{lemma}\label{lemma:witness-binary}
  Every binary regular query $\triples[x,y]$ has witness functions of bounded guidance width.
\end{lemma}
\begin{proof}
Whenever a pair $(x,y)$ belongs to $\triples(t)$, call the node $z = \gca(x,y)$ an
$x$-intermediate node, and call $y$ a $z$-final node.
We will define two guidance systems, the first one directing any source node $x$ to an
$x$-intermediate one, and the second one directing any intermediate node $z$ to a $z$-final one. 
The two guidance systems will be combined using Lemma~\ref{lemma:comp-guid-systems}.

A binary query is essentially a degenerate case of ternary query,
with $y_1 = y_2$.
By Lemma~\ref{lemma:path-based} assume that $\triples$ is path-based.
Thus its truth value in a tree $t$ only depends on some regular properties 
of two words 
\[
w_1 = \word_t(x, \gca(x,y)) \ \text{ and } \ w_2 = \word_t(\gca(x,y), y),
\]
 as depicted in the figure in Section~\ref{sec:path-based-queries}.
(Words $w_3$ and $w_4$ are empty as $y_1 = y_2$.)
Namely, $(x, y)$ belongs to $\triples(t)$ if and only if $(\alpha(w_1), \alpha(w_2)) \in F$,
for a designated set $F \subseteq S^2$. Fix $(s_1, s_2) \in F$. We will define a guidance system for
pairs $(x, y)$ which satisfy
\[
\alpha(w_1) = s_1 , \ \ \alpha(w_2) = s_2.
\] 
(Then the required guidance system will be a disjoint union over all pairs $(s_1, s_2) \in F$.)

Assume further, without loss of generality, that $x$ is in the left subtree, and $y$ is in the right subtree
of $\gca(x,y)$, including possibly $y = \gca(x,y)$. 
(Again, the required guidance system will be a disjoint union of two systems.)

Consider deterministic word automata $\aut_1$ and $\aut_2$ that recognize the properties $\alpha^{-1}(s_1)$ and
$\alpha^{-1}(s_2)$, respectively.
Think of a run of $\aut_1$, starting from a source node $x$, 
along the path in $t$ leading from $x$ to some $x$-intermediate node.
Consider such runs of $\aut_1$ starting from all source nodes $x$, one run from every source node.
These runs may be translated into a guidance system, as follows.

Each of the runs labels nodes on the path from $x$ to an $x$-intermediate node with states of $\aut_1$.
The idea is that two source nodes $x$ and $x'$ may be directed to the same intermediate node if the two runs of
$\aut_1$ that start in $x$ and $x'$ label some node of $t$ with the same state. 
In other words, one may use the same guide both for $x$ and $x'$.
Thus there is a guidance system, with as many colors as the number of states of $\aut_1$,
that follows the runs of $\aut_1$ until acceptance, and directs any source node $x$ to 
some $x$-intermediate node.

Likewise for $\aut_2$, there is a corresponding guidance system, that leads any intermediate
node $z$ to some $z$-final node. 
Applying Lemma~\ref{lemma:comp-guid-systems} for these two guidance systems we get the result.
\end{proof}

\subsection{Arrangement \eqref{eq:a1}} In this section we
show that Theorem~\ref{thm:bounded-width} holds if all
triples selected by $\triples[x,y_1,y_2]$ have
arrangement~\eqref{eq:a1}, pictured below.
\begin{center}
  \includegraphics[scale=0.65]{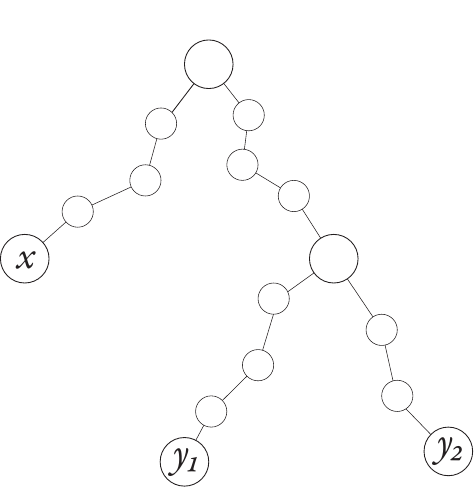}
\end{center}

Suppose $\tau[x,y]$ is a binary query, and $\sigma[y,y_1,y_2]$ is a
ternary query. We define the following ternary query
\begin{eqnarray*}
  \tau \circ_y \sigma[x,y_1,y_2] \qquad = \qquad  \exists y \ \tau[x,y]
  \land \sigma[y,y_1,y_2]\ .
\end{eqnarray*}

\begin{lemma}\label{lemma:witness-decompose}
%
%
 Let $\tau,\sigma$ be as above. If $\tau$ and $\sigma$ have bounded width
 witness functions then so does
 $\tau \circ_y\sigma$.
\end{lemma}
\begin{proof}
  By considering the witness function for $\tau \circ_y \sigma$ obtained as
  a composition and applying Lemma~\ref{lemma:comp-guid-systems}.
%
\end{proof}

\begin{lemma}\label{lemma:decompose-gca}
  We may assume without loss of generality that $\triples$ only selects triples $(x,y_1,y_2)$ where $x=\gca(y_1,y_2)$.
\end{lemma}
\begin{proof}
By the considerations in Section~\ref{sec:path-based-queries}, we know
that a triple $(x,y_1,y_2)$ is selected by $\triples$ if and only if the
images, under the morphism $\alpha$, of the four path words
$w_1,w_2,w_3,w_4$ belong to a designated set $F \subseteq S^4$ of
accepting tuples.

Let $s_1,\ldots,s_4\in S$. Let $\tau_{s_1,s_2}$  be the binary query that selects a
pair $(x,y)$ if
\begin{align*}
\alpha( \word_t(\gca(x,y),x)) & = s_1 \\ 
\alpha( \word_t(\gca(x,y),y)) & = s_2\ .
\end{align*}
Likewise, let $\sigma_{s_3,s_4}$ be the ternary query that selects a
triple $(y,y_1,y_2)$ if 
\begin{align*}
\gca(y_1,y_2) & = y \\ 
\alpha( \word_t(y,y_1)) & = s_3 \\ 
\alpha( \word_t(y,y_2)) & = s_4\ .
\end{align*}
The queries $\tau_{s_1,s_2}$ and $\sigma_{s_3,s_4}$ 
%
%
can be joined to define $\triples$, in the following way.
\begin{align*}
  \triples = \bigcup_{(s_1,s_2,s_3,s_4) \in F} \tau_{s_1,s_2} \circ_y \sigma_{s_3,s_4} .
\end{align*}
By  Lemma
\ref{lemma:witness-decompose}, we see that the width of  witness
functions for $\triples$ is bounded by the widths of the witness
functions for the $\tau$ queries, 
which is bounded by Lemma~\ref{lemma:witness-binary}, and
the width of the witness functions for the $\sigma$ queries.  The
latter are queries where the first variable is the $\gca$ of the
second and third variables, which concludes the proof of the lemma.
\end{proof}

Thanks to the above lemma, we are left with a query $\triples$ that
selects triples in the arrangement pictured below
(for future reference let us call this arrangement \emph{trivial}).
\begin{center}
  \includegraphics[scale=0.65]{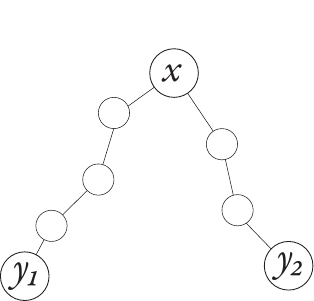}
\end{center}
We will provide a 2-color guidance system that induces 
a witness function for $\triples$ in $(t,E)$.
This is guaranteed by the following lemma:

\begin{lemma}
\label{lem:trivial-arr}
Let $\triples$ be any (not necessarily regular) query that selects
only nodes in a trivial arrangement.
Then $\triples$ has witness functions of guidance width 2.
\end{lemma}
\begin{proof}
The guidance system is constructed in a single root-to-leaf pass.

More formally, for each set $X$ of nodes that is closed under
ancestors, we will provide a guidance system $\Pi_X$ that directs each
node that is a source node and in $X$ to some witness, either $y_1$ or $y_2$.  The guidance
system will have the additional property that no tree edge is traversed by two guides.

The guidance system is constructed by induction on the size of
$X$. The induction base, when $X$ has no nodes, is straightforward. We
now show how $\Pi_X$ should be modified when adding a single $x$ node
to $X$.  When $x$ is not a source node, then nothing needs to be done. Otherwise, suppose that $x$ is a source node, and the witness is $(y_1,y_2)$.  Since all guides in $\Pi_X$ originate in nodes from $X$, any
guide that passes through $x$ must also pass through its parent. Using
the additional assumption, we conclude that at most one guide $\pi$
from $\Pi_X$ passes through $x$. In particular, either the left subtree of $x$, which contains $y_1$,   or the 
right subtree of $x$, which contains $y_2$,  has no guide passing through it.  We create a new guide  that connects $x$ to the witness in the subtree without a guide.
\end{proof}

For arrangement~\eqref{eq:a1} the proof of Theorem~\ref{thm:bounded-width}
is thus completed.

\subsection{Arrangements~\eqref{eq:a2} and~\eqref{eq:a3}}
\label{sec:arrangm-eqref-eqref}

For the remaining arrangements, in this section we only show how 
they can be reduced to the simplified ones.
We formulate Theorem~\ref{thm:simplified-bounded-width} below, which we will
use in this section, and which follows easily from the Main Lemma
(forthcoming Lemma~\ref{l:main}) to be proven in the next section.

To state the theorem, we need a new notion. 
A guidance system in a graph tree is called 
\emph{consistent} wrt.~a given ternary query
if each of guides obeys the following uniqueness requirement: 
whenever a set $Z \subseteq X$ of source nodes is guided to the same node $y$,
then there is a pair $(y_1, y_2)$ that is a witness pair for all
nodes $Z$, with $y_1=y$ or $y_2=y$.
Roughly speaking: 
if all nodes in $Z$ agree on the witness they are guided to,
then they agree on the other witness as well.
The notion of consistency is  meaningful only
relative to a given ternary query.
Below, the consistency property will make it possible to combine two
guidance systems appropriately.

\begin{theorem}
\label{thm:simplified-bounded-width}
Every simplified regular query has witness functions of bounded guidance width.
Furthermore, a consistent guidance system always exists (of the required bounded guidance width).
\end{theorem}

Now using Theorem~\ref{thm:simplified-bounded-width} we prove
Theorem~\ref{thm:bounded-width} for arrangements~\eqref{eq:a2} and~\eqref{eq:a3}.
By symmetry, we only consider the arrangement~\eqref{eq:a2}.  
We simplify the arrangement in two steps. First we claim that without loss of
generality $x$ can be assumed to be an ancestor of $y_2$ -- this may be shown
by essentially the same technique as in Lemma~\ref{lemma:decompose-gca}
hence we omit the details.
Second, we
show that $y_1$ can be assumed to be an ancestor of $x$ and $y_2$. The
arrangement~\eqref{eq:a2}, as well as its two successive
simplifications, are pictured below.
\begin{center}
  \includegraphics[scale=0.5]{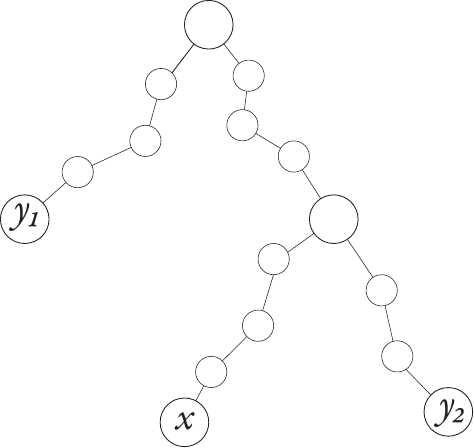} \qquad
  \includegraphics[scale=0.5]{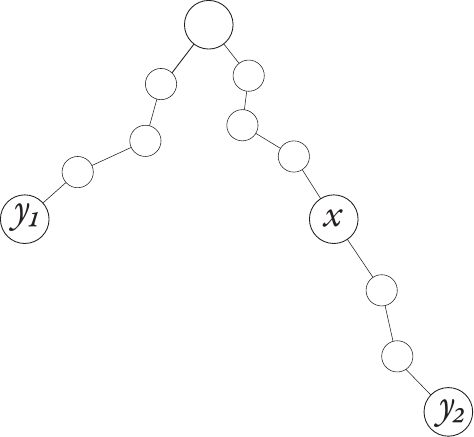} \qquad 
  \includegraphics[scale=0.5]{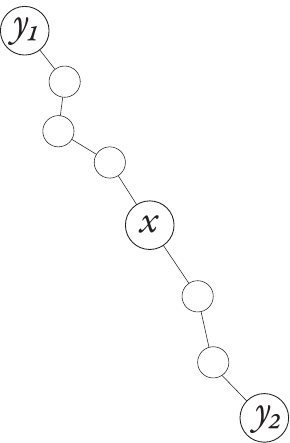}
\end{center}

Let our starting arrangement  be the middle one in the picture above,
i.e.~we assume that the first simplification has been already applied.
Without loss of generality we may assume that $y_1$ is in the left subtree, and $x$ in the right subtree
of the $\gca(x, y_1)$ node (thus we again split into two sub-cases), 
and that both $x$ and $y_1$ are not equal to $\gca(x, y_1)$.

By the considerations in Section~\ref{sec:path-based-queries}, we know
that a triple $(x,y_1,y_2)$ is selected by $\triples$ in $t$ if and only if the
images, under the morphism $\alpha$, of the three path words:
\[
\word_t(x, \gca(x, y_1)), \ 
\word_t(x, y_2), \ 
 \word_t(\gca(x, y_1), y_1),
\]
belong to a designated set $F \subseteq S^3$ of accepting tuples.

Fix $(s_1,s_2,s_3) \in F$. We define a guidance system for triples $(x,y_1,y_2)$ which satisfy 
\[
s_1= \word_t(x, \gca(x, y_1)), \ 
s_2= \word_t(x, y_2), \ 
s_3 =  \word_t(\gca(x, y_1), y_1).
\]
In general,  the guidance system will be a disjoint
union over all triples $(s_1, s_2, s_3) \in F$. 
Consider an arbitrary graph tree $(t, E)$ over which $\triples$ is evaluated, together with an arbitrary
subset $X$ of source nodes.
We aim at constructing a guidance system whose number of colors depends only on $\triples$,
that directs any source node $x \in X$ to either $y_1$ or $y_2$, for some triple $(x, y_1, y_2)$ selected
by $\triples$.
We will do it in two stages. In the first stage, the node $x$ is directed either to $y_2$, or
to $y = \gca(x, y_1)$. By Theorem~\ref{thm:simplified-bounded-width} we will be able to assume that this guidance
system is consistent. In the second stage, every $y$ node will be directed either to
an appropriate $y_1$ node, or to the $y_2$ node, using Lemma~\ref{lem:trivial-arr}.
Finally, we will compose the two guidance systems using Lemma~\ref{lemma:comp-guid-systems}.

Formally speaking, for the first stage we use the simplified query
$\sigma_{s_1,s_2}[x, y, y_2]$  that selects a triple $(x,y,y_2)$ if 

\begin{iteMize}{$\bullet$}
\item $\alpha(\word_t(x,y)) = s_1$ 
\item $\alpha(\word_t(x,y_2)) = s_2$ 
\item $y < x \leq y_2$
\item
$x$ is in the right subtree of $y$.
\end{iteMize}

The idea now is that the query $\sigma_{s_1,s_2}$ is evaluated over a modified graph tree
$(t, E_{s_3})$. The relation $E_{s_3}$ is defined as follows: $(y, y_2)$ is in $E_{s_3}$ iff
$(y_1, y_2) \in E$ for some $y_1$ such that

\begin{iteMize}{$\bullet$}
\item
$y = \gca(y_1, y_2)$,
\item
$y_1$ is in the left subtree of $y$, 
\item
$\alpha(\word_t(y, y_1)) = s_3$.
\end{iteMize}

Intuitively, the first node $y_1$ of every edge $(y_1, y_2) \in E$ is moved
to $y = \gca(y_1, y_2)$, but only if the equation $\alpha(\word_t(y, y_1)) = s_3$ holds.
Note that we  use here  the more general notion of  graph trees, rather than data trees,
as the relation $E_{s_3}$ is not an equivalence in general.
By Theorem~\ref{thm:simplified-bounded-width} we know that $\sigma_{s_1,s_2}$
has a witness function induced
by a consistent $m$-color guidance system $\Pi$, where $m$ only depends on $\sigma_{s_2, s_3}$ and does not depend
on $t$, $E$ or $X$.

Each source node $x$ is directed so far either to some $y < x$, or to $y_2 \geq x$.
Without loss of generality we may assume that colors used for target nodes of the first kind
are distinct from colors used for target nodes of the other kind.
For the second stage, consider the set $Y$ of all target nodes $y$ of $\Pi$ of the first kind,
and restrict attention to the induced sub-guidance system of $\Pi$.
Note that every such node $y \in Y$ has an associated node $y_1$ located in the left subtree of $y$
such that $\alpha(\word_t(y, y_1)) = s_3$. Moreover, due to consistency of $\Pi$,
$y$ has also an associated node $y_2$ located in the right subtree of $y$ such that
$\alpha(\word(y, y_2)) = s_1 s_2$.
Consider the set of all such triples $(y, y_1, y_2)$ and apply Lemma~\ref{lem:trivial-arr}
to obtain a 2-color guidance system that directs every $y \in Y$ to one of its two associated nodes.

Finally, using  Lemma~\ref{lemma:comp-guid-systems} we obtain a guidance system for $\triples$ of
bounded width.

\removed{
Let $\tau_{s_3, \Pi}[y, y_1, y_2]$  be a ternary query that selects
a triple $(y, y_1, y_2)$ in $(t, E)$ if
\begin{iteMize}
\item $y = \gca(y_1, y_2)$, $y_1$ is in the left subtree of $y$ and $y_2$ in the right one,
\item $y$ is one of the target nodes of $\Pi$,
\item $(y_1, y_2) \in E$ is a witness pair for \emph{all} nodes $x$ that are guided 
by $\Pi$ to $y$ (by consistency of $\Pi$, such a pair $(y_1, y_2)$ exists for any
target node $y$ of $\Pi$),
\item $\alpha(\word_t(y, y_1) = s_3$.
\end{iteMize}
Query $\tau_{s_3, \Pi}$ depends on the guidance system $\Pi$, we thus implicitly assume
that $\Pi$ is included in the labeling of a tree.
Query $\tau_{s_3, \Pi}$ is not a regular one in general, but it selects only triples in
 trivial arrangement.
 Applying Lemma~\ref{lem:trivial-arr} we get a 2-color guidance
 system that induces a witness function for $\tau_{s_3, \Pi}$.
 The two guidance systems can be then combined into one $4m$-color guidance
 system due to Lemma~\ref{lemma:comp-guid-systems}.
This guidance system induces a witness function for $\triples$ and $X$ in $(t, E)$.
}
 
  As the graph tree $(t, E)$ and the subset $X$ were chosen arbitrarily, 
this completes the proof  of Theorem~\ref{thm:bounded-width}.


%% file: updown.tex
\section{Proof of Theorem~\ref{thm:simplified-bounded-width}}
\label{sec:diff-arrang-trees}

Fix in this section a simplified regular query $\triples[x,y_1,y_2]$, i.e., satisfying
the conditions~\eqref{en:simplify1} and~\eqref{en:simplify2} from
Section~\ref{sec:simplifying-delta}.
Since  the query is regular, the dependency stated in item~\eqref{en:simplify2} is a regular dependency. 
We may thus assume a semigroup morphism $\alpha : \Sigma^* \to S$ recognizing $\triples$, which maps each word to an element of a finite semigroup $S$. 
Whether or not a triple $(x,y_1,y_2)$ is selected by $\triples$ depends only on the images 
\begin{eqnarray}
\label{e:infixes}
  s_1 &=& \alpha(\word_t(y_1, x)) \in S \\
  s_2 &=& \alpha(\word_t(x, y_2)) \in S .
\end{eqnarray}
In other words, there is a set of \emph{accepting pairs} $F \subseteq
S^2$ such that $\triples(w)$ is the set of triples $(x,y_1,y_2)$ with
$(s_1, s_2) \in F$.
We fix the morphism $\alpha$ for the rest of this section.

We distinguish two types of edges in a graph tree $(t,E)$. The \emph{tree edges} are edges that connect parents  with children, as well as a dummy edge going into the root of the tree and dummy edges going out of the leaves. The \emph{class edges} are the edges from $E$.
We order tree edges by the ancestor relation $\leq$, according to the positions
in the tree, with the dummy edges coming as the least one and the maximal ones, 
respectively. For two tree edges $e \leq f$ in a tree $t$, we write
$\word_t(e, f)$ for the word labeling $t$ on the path that begins in the target of $e$
and ends in the source of $f$.  In particular $\word_t(e, e) = \epsilon$.

\sredniparagraph{Forward Ramseyan splits} The key tool in our proof is a forward Ramseyan split, as defined by Colcombet in~\cite{DBLP:conf/icalp/Colcombet07}.
Let $t$ be a tree labeled with $\Sigma$.
A \emph{split of height $n$} in $t$ is a function $\sigma$ that
maps each tree edge to a number in $\set{1, \ldots, n}$.  
We say that two tree edges $e < f$ are \emph{neighbors} with
respect to a split $\sigma$, if $\sigma$ assigns the same number to
$e$ and $f$, and all tree edges between $e$ and $f$ are mapped to at most
$\sigma(e)$.   A split $\sigma$ is called \emph{forward Ramseyan} with
respect to a morphism $\alpha$ if
\begin{equation}
\label{e:forwardRams}
\alpha(\word_t(e, f)) = \alpha(\word_t(e, g))
\end{equation}
holds for every three pairwise neighboring tree edges $e < f < g$.
The following theorem was shown  in~\cite{DBLP:conf/icalp/Colcombet07}.   

\begin{theorem}\label{thm:ramseyan-split}
  Fix a morphism $\alpha : \Sigma^* \to \monoid$. 
  Every tree $t$ has a forward Ramseyan split of height $\Oo(|S|)$. 
 Furthermore, the split is top-down deterministic in
 the sense that all tree edges from a node to its children are assigned the same number in the split.
\end{theorem}



From now on wlog we consider only \emph{complete} binary trees, where each non-leaf node
has precisely two children.

\sredniparagraph{Factors.} 
Two comparable wrt.~$\leq$ (i.e., belonging to one path) tree edges $e$ and $f$ are called \emph{visible}
if all tree edges between $e$ and $f$ are mapped by the split to values
strictly smaller than $\sigma(e)$ and $\sigma(f)$.  Visible pairs of tree edges
naturally determine a nested factorization of $t$ in the following
way. 

A \emph{pre-factor} in a tree $t$ is a connected set of nodes
(connected by tree edges) such that if a node $x$
is in the pre-factor, then either all children of $x$ are in the
pre-factor, or none of them.  Each pre-factor of $t$ has a root and some leaves (maximal nodes wrt.~$\leq$), and inherits its edges from
$t$. In the definitions below, we talk about tree edges and not class edges.
We distinguish \emph{internal} edges of a pre-factor,
connecting two nodes in that pre-factor,
and \emph{external} edges connecting the root or the leaves with some node outside the pre-factor.
This includes the tree edge leading to the root of the pre-factor 
(called the \emph{root edge} of the pre-factor) and the tree edges going out of 
the leaves (called the \emph{leaf} edges of the pre-factor).
Note that external edges may be either proper tree edges, or dummy edges.
As the split $\sigma$ is assumed to be deterministic, all tree edges leaving a given
leaf of a pre-factor are assigned the same number. A pre-factor $F$ is
called a \emph{factor} in $t$ if it respects the split $\sigma$ in the
following way: the root edge is visible from each of the leaf edges.  This
means that on each (shortest) path in a factor from its root edge to a leaf edge,
numbers assigned by $\sigma$ to the internal edges on that path
are strictly smaller than those assigned to the two external edges.
By the height of a factor we mean the greatest number assigned to an
internal edge, or $0$ if no such edge exists (the case of one-node
pre-factor).  Additionally, the whole tree $t$ is also a factor if,
wlog, we assume that the root dummy edge is visible from all leaf dummy edges; 
its height at most equals the height of $\sigma$.

A \emph{subfactor} of a factor $F$ is any factor $G
\subsetneq F$ that is maximal with respect to inclusion. By the
definition of factor, we get:
\begin{claim}\label{claim:disj-subf}
Every two different subfactors of $F$ are 
disjoint (have disjoint sets of nodes, but possibly share an external edge).
\end{claim}
\begin{proof}
Indeed, assume two non-disjoint different subfactors $F_1, F_2$ of some factor $F$.
The root node of one of them, say the root node $r_1$ of $F_1$, is necessarily contained in the other.
As $F_1$ is not included in $F_2$, there must be a leaf node $l_1$ of $F_1$ not contained
in $F_2$; the path that leads to that leaf node passes though a leaf node $l_2$ of
$F_2$. If we denote by $r_2$ the root node of $F_2$ we know that the four nodes
are located on one path in the following order:
\[
r_2 \leq r_1 \leq l_2 < l_1 .
\]
This arrangement of the nodes is in a clear contradiction with the assumption that
the tree edge incoming to $r_i$ is visible from (all) tree edges outgoing from $l_i$,
for $i = 1, 2$.
\end{proof}
Hence each factor $F$ is the disjoint union of its subfactors.
We say a subfactor $G$ is an \emph{ancestor} of a subfactor $H$ if
their roots are so related. Likewise we talk about a subfactor being a
\emph{child} or \emph{parent} of some other subfactor.

A factor of height 2 together with its decomposition into subfactors is pictured below.

\begin{center}
  \includegraphics[scale=1]{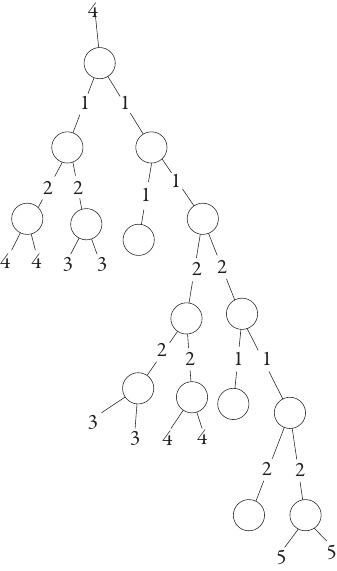}
  \includegraphics[scale=1]{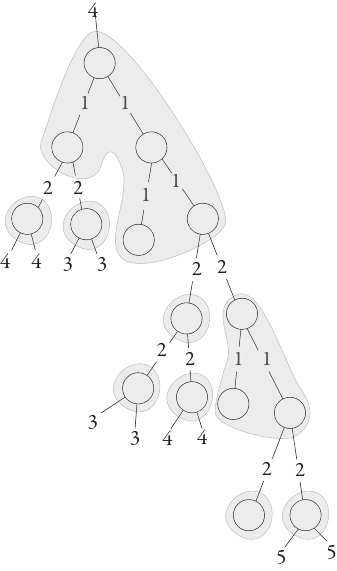}
\end{center}

Our proof of Theorem~\ref{thm:simplified-bounded-width} is based on the Main Lemma stated below. The lemma is proved by induction on the height of factors.
To state the lemma recall the notion of consistent guidance system introduced in Section~\ref{sec:arrangm-eqref-eqref}.

\begin{lemma}[Main Lemma]
\label{l:main}
  Fix a factor height $h$.
  There is a bound $n \in \Nat$, depending only on $\triples$ and $h$,
   such that for every graph tree
  $(t,E)$, every factor $F$ in $t$ of height $h$, and every set $X \subseteq F$ of
  source nodes, there is a witness function for $\triples$ and $X$ in $(t, E)$
  induced by a consistent guidance system using at most $n$ colors. 
  Furthermore, this witness function only points to
  descendants of the root of $F$.
\end{lemma}

The proof of the lemma is by induction on the height $h$.
The number of colors $n$ will depend on $h$
and the size of the monoid $S$ recognizing the query. It will
not depend on $t$. When going from height $h$ to height $h+1$, there
will be a quadratic blowup in the number of colors. Therefore, $n$
will be doubly exponential in the height of $F$. 

Since the witness function will be induced by a guidance system, the
last assumption in Lemma~\ref{l:main} could be
restated as saying that no guide passes through the root edge of $F$.
Theorem~\ref{thm:simplified-bounded-width} is a special case of the Main
Lemma when $F$ is the whole tree.

The base case when $h = 0$, and hence the factor $F$ has one or no nodes, is easy ($1$
color, going downwards, is sufficient).  For the induction step, fix a
factor $F$, and assume that there is a bound $n$ sufficient for any
factor of smaller height than $F$, which includes all subfactors of $F$.  
Below, subfactors of $F$ are simply called \emph{subfactors}, 
without explicitly referring to $F$.

A tree edge of $F$ that is an external edge of one if its subfactors is
called a \emph{border edge}.  In particular each external edge of $F$
is a border edge.  
Special care will be paid in our proof to internal (i.e.~not external)
border edges, i.e., the edges that connect one subfactor to another.

\begin{claim}\label{claim:visible}
If two internal border edges in a factor are comparable by the ancestor relation $\leq$
then they are assigned the same value by the split.
\end{claim}
\begin{proof}
Assume two internal border edges $e < f$, with different values assigned by the split,
such that no other internal border edge is located on the (shortest) path from $e$ to $f$. 
As values assigned to all external border edges are strictly larger, 
one of $e, f$ is visible from some (possibly external) border edge ``over'' the other.
That is, either $e$ is visible from some $e' > f$, or $f$ is visible from some $f' < e$. 
In both cases, one of $e$, $f$ is an internal edge of some subfactor -- a contradiction.
\end{proof}

We do a case distinction, regarding the
number of internal border edges on the paths from a source to witness
nodes.  For a node $x \in X$ and a witness $(y_1,y_2)$ we define two
numbers $m_1,m_2$. Let $m_1$ be the number of internal
border edges on the path between $y_1$ and $x$, and let $m_2$ be the
number of internal border edges on the path between $x$ and $y_2$. For
technical convenience, we deliberately choose not to count external
border edges.  We divide the set $X$ into three parts:
\begin{enumerate}
\item[$X_1$\ \ ] Nodes $x \in X$ that have a witness with $m_2 \le 1$.
\item[$X_2$\ \ ] Nodes $x \in X$ that have a witness with $m_1 \le 1$ and $m_2 \ge 2$.
\label{en:ancestor}
\item[$X_3$\ \ ] Nodes $x \in X$ that have a witness with $m_1,m_2 \ge 2$.
\end{enumerate}
We prove the Main Lemma for each of the three parts separately. Next,
we combine the three guidance systems into a single guidance system.
Our construction will yield two kinds of guides:
the \emph{ancestor guides} pointing to the first witness and thus going up a tree;
and \emph{descendant guides} pointing to the second witness, and thus going down the tree.
Interestingly, ancestor guides will be only created in case of nodes from $X_2$.    
All the guides will satisfy the consistency condition required in 
Lemma~\ref{l:main}.

\paragraph*{\bf Nodes from $X_1$, i.e.~nodes  that have a witness with $m_2 \le 1$.} 
Consider a subfactor $G$ of $F$.
In this case,  each node $x \in X_1 \cap G$ has a descendant witness $y_2$ that is either in $G$, or in a  child subfactor of $G$, or perhaps outside $F$. Apply the induction assumption to  $G$, producing a guidance system $\Pi_G$ with at most $n$
colors.  Since the Main Lemma requires the guidance system to point to
descendants of the factor's root, and $m_2 \le 1$, we infer that inside $F$ the
guides of $\Pi_G$ can only intersect $G$ and its child subfactors, and
no other subfactors (it is possible that the guides leave the factor
$F$, though).  Therefore, all the guidance systems $\Pi_G$, for all subfactors $G$ of $F$, can be
combined into a single guidance system with at most $2n$ colors, used
alternatingly for even and odd depths.

\paragraph*{\bf Nodes from $X_2$, i.e.~nodes   that have a witness with $m_1 \le 1$ and $m_2
  \ge 2$.} 
In this case, for each node $x \in X_2$ there is an ancestor witness
$y_1$ that is either in the subfactor of $x$, in the parent subfactor,
or outside $F$.
Note that the latter is possible only when $x$ belongs 
either to the root subfactor of $F$, or to some of its child subfactors;
denote this set of subfactors by $\Gg_0$.
We will construct the guidance system in a step-by-step manner, for all
subfactors, according to the ancestor ordering.

Formally speaking, consider a family $\Gg$ of subfactors that is closed
under ancestors and includes $\Gg_0$.   
We provide a guidance system $\Pi_\Gg$ of $4n^2+3n$
colors that provides witnesses for all nodes of $X$ belonging to the
subfactors in $\Gg$.
The construction of $\Pi_\Gg$ is by induction on the number of
subfactors in $\Gg$. 

The induction base is when $\Gg$ equals $\Gg_0$.
For each subfactor $G \in \Gg_0$, we apply the  Main Lemma, for the smaller height, to $G$ and nodes from $X_2$
that belong to $G$, yielding an $n$-color guidance system. 
We combine these guidance systems into $\Pi_{\Gg_0}$ as follows: 
use one set of $n$ colors for the root subfactor, and another set of $n$ colors 
for all the child subfactors of the root.

   For the induction step, suppose that we have already
constructed $\Pi_\Gg$ for $\Gg$, and that $G \not \in \Gg$ is a
subfactor whose parent is in $\Gg$.  Consider the guides of $\Pi_\Gg$
that pass through the root edge of $G$.  We apply two
distinctions to these guides. First, we use the name \emph{parent
  guides}, for the guides that originate in the parent subfactor of $G$,
and the name \emph{far guides} for the other guides. Second, we use
the name \emph{ending guides} for the guides whose target is in $G$
and the name \emph{transit guides} for the other guides, which
continue into a child subfactor of $G$, or even exit $F$. Altogether, there are four
possibilities: parent transit guides, far ending guides, etc. We
assume additionally that there are at most $n$ parent guides and at
most $2n$ far and parent guides altogether, 
and hence at most $2n$ guides entering $G$. This additional invariant is satisfied by the induction base, and it will be preserved through the construction.


Apply  the induction assumption of the Main Lemma to $G$ and nodes from $X_2$
that belong to $G$, yielding
a guidance system $\Pi$ with $n$ fresh colors. We use the name
\emph{starting guides} for the guides of $\Pi$.  We want to combine
$\Pi_\Gg$ with $\Pi$ in such a way that the resulting guidance system
still uses at most $4n^2+3n$ colors,  like $\Pi_\Gg$, and satisfies the additional invariant. If we were
to simply take the two systems together, we might end up with a leaf edge of $G$ which is traversed both by starting and
transit guides, which could exceed the bound $2n$ on guides passing
through border edges.

We solve this problem as follows. Consider the  leaf edges of
$G$ that are traversed by the far transit guides. There are at most
$2n$ such edges by our invariant assumption.  We will remove all starting guides
that pass any of these edges, and find other witnesses for nodes that use these
starting guides.   
This guarantees that the invariant condition is recovered: 
at most $2n$ guides passes through any  leaf edge of $G$,
and at most $n$ of them are starting guides. These other witnesses will be ancestors. This explains why the induction  starts with $\Gg$ containing the root subfactor and its children, since these are the subfactors that may have ancestor witnesses outside the whole factor $F$ (recall that passing through the root edge is not counted in $m_1$). The statement of the Main Lemma does not allow guides that pass through the root edge of $F$.

The removing of starting guides proceeds as follows.
Let $e$ be a leaf edge of
$G$ that is traversed by a far transit guide, which has color $j$
in the guidance system $\Pi_\Gg$. Let $\pi$ be a starting guide, which
has color $i$ in $\Pi$, that also traverses $e$, with $y_2$ its target
node.  By the consistency property of $\pi$, there is some $y_1$ such that
$(y_1,y_2) \in E$ is a witness pair for all source nodes of $\pi$. Note that by
assumption on $m_1 \le 1$, the node $y_1$ is either in the subfactor
$G$ or its parent. We create an ancestor guide with a fresh color that
connects all the source nodes of $\pi$ to $y_1$.  The color of this
guide, which we call an \emph{ancestor color}, will take into account
three parameters: the colors $i$ and $j$, as well as a parity bit $b
\in \set{0,1}$. The parity bit is $0$ if and only if $G$ has an even
number of ancestor subfactors. We use the triple $(i,j,b)$ for the
color name.

We will show that this new ancestor guide does not conflict with any
other ancestor guide with the same color. Each new ancestor guide is
contained in $G$ and possibly its parent subfactor $H$, by assumption
on $m_1 \le 1$.  Inside the subfactor $G$ there is at most one
ancestor guide of each color, so there are no collisions inside
$G$. One could imagine, though, that a new ancestor guide $\pi$ with
color $(i,j,b)$ collides inside the parent subfactor $H$ with some
other ancestor guide $\pi'$ of the same color.  Since the colors of
$\pi$ and $\pi'$ agree on the parity bit $b$, we conclude that the
guide $\pi'$ cannot originate in $H$, which has a different parity
than $G$. Therefore, $\pi'$ must originate in some other child
subfactor of $H$, call it $G'$, that had been previously added to
$\Gg$. Since the color of $\pi'$ is also $(i,j,b)$, we conclude that
$j$ was the color of a far transit guide in $G'$. This is impossible,
since a far transit guide in $G'$ or $G$ must originate not in $H$,
but in an ancestor of $H$, and therefore there would be a collision in
the root of $H$.

The ancestor guides created above are the only ancestor guides in our solution.
In the subfactors where they are used, the ancestor guides have the target in
the parent subfactor.  
\removed{
These ancestor guides are not used in the root subfactor $F_{\min}$,
 since the
root subfactor is not traversed by any transfer guides. 
Moreover, they are not used in the child subfactors of $F_{\min}$,
as the ancestor guides could exit $F$ in this case (and thus violate the requirement of the Main Lemma).
Thus $\cal G$ consisting of the factor $F_{\min}$ and all its child subfactors constitutes 
the induction base;
this choice supports the invariant condition: 
at most $n$ parent guides and at most $n$ far guides enter a subfactor $G$.
}

Let us count the number of colors used. We need $2n$ colors for the
transit guides, and $n$ colors for the starting guides. For the
ancestor guides, we need $4n^2$ colors.  Altogether, we need $4n^2+3n$
colors.
Note that all the guides satisfy the consistency condition required by
Lemma~\ref{l:main}.

\paragraph{\bf Nodes from $X_3$, i.e.~nodes   that have a witness with $m_1,m_2 \ge 2$.}
In this case, each node $x \in X_3$ has a witness $(y_1,y_2)$ such that
the path from $y_1$ to $x$, as well as the path from $x$ to $y_2$,
passes through at least two internal border edges.  This case is the
only one where we use the forward Ramseyan split. 

Consider a source $x$ with a witness $(y_1,y_2)$.  The internal border
edges naturally split $\word_t(y_1, x)$ and $\word_t(x, y_2)$ into $m_1
{+} 1$ and $m_2 {+} 1$ words, respectively:
\[
\begin{array}{rcl}
\word_t(y_1, x) & = & v_0 {\cdot} v_1 {\cdot} \ldots {\cdot} v_{m_1} \\
\word_t(x, y_2) & = & w_0 {\cdot} w_1 {\cdot} \ldots {\cdot} w_{m_2}.
\end{array}
\]
The first letter of $v_0$ is the label of $y_1$. The last letter of
$v_{m_1}$ and also the first letter of $w_0$ is the label of $x$. The
last letter of $w_{m_2}$ is the label of $y_2$ (recall that 
$y_1$ or $y_2$ might be  outside $F$).  Furthermore, each two consecutive internal border edges are not
only visible, but also neighbouring, by  Claim~\ref{claim:visible}.  Hence, 
as we have a forward Ramseyan split (cf.~\eqref{e:forwardRams}), the
values $\alpha(\word_t(y_1, x))$ and $\alpha(\word_t(x, y_2))$ are determined by
the first two parts and the last part:
\begin{equation}
\label{eq:forwApplied}
\begin{array}{lrcl}
\mathrm{(i)}  & 
\alpha(\word_t(y_1, x)) & = & \alpha(v_0) {\cdot} \alpha(v_1) {\cdot} \alpha(v_{m_1})\\
\mathrm{(ii)} &
\alpha(\word_t(x, y_2)) & = & \alpha(w_0) {\cdot} \alpha(w_1) {\cdot} \alpha(w_{m_2}).
\end{array}
\end{equation}
Let us fix six values $s_1,\ldots,s_6 \in S$. By splitting the set $X_3$
into at most $\size{S}^6$ parts, each requiring a separate
guidance system, we can assume that each $x \in X_3$ has a witness where
\begin{align*}
  \begin{array}{lll}
    s_1=\alpha(v_0)\quad&  s_2=\alpha(v_1)\quad& s_3=\alpha(v_{m_1}) \\
s_4=\alpha(w_0) & s_5=\alpha(w_1)& s_6=\alpha(w_{m_2}) .
  \end{array}
\end{align*}
We will only consider witnesses that satisfy the assumptions above.


We now proceed to create the guidance system. As in the case $m_1 \le
1$, the guidance system will be defined for a family $\Gg$ of
subfactors that is closed under ancestors.  The guidance system will
use at most $3$ colors, and will have the following additional
invariant property: if $e$ is an edge that connects a subfactor $G$ with a child
subfactor $H$, then at most two guides pass through
$e$. Furthermore, if exactly two guides pass through $e$, then one
of the guides has its target in $H$.

The construction is by induction on the number of subfactors in
$\Gg$. The induction base when $\Gg$ has no subfactors is
obvious. Below we show how to modify a guidance system $\Pi_\Gg$ for $\Gg$
when adding a new subfactor $G$.

Consider the (at most two) guides of $\Pi_\Gg$ that pass through
the edge connecting $G$ to $\Pi_\Gg$. As in the case $m_1 \le 1$, we
use the term \emph{transit guide} for the guides of $\Pi_\Gg$ that
enter $G$ through its root and exit through one of its external leaf
edges. By the invariant assumption, there is at most one transit guide.

We now define a guidance system $\Pi$ for the nodes in $G$, which we will next combine with $\Pi_\Gg$.
\begin{claim}
	There is a one-color guidance system $\Pi$ defining a witness function for all nodes in $G \cap X_3$.
\end{claim}
\begin{proof}
 For each node $x \in G \cap X_3$, choose the lexicographically first witness $y_2 > x$ that satisfies the assumptions on the six images in the semigroup, call it $y_x$. Let $Y$ be all these witnesses $y_x$, for $x \in G \cap X_3$; this set is an antichain with respect to the descendant relation. For each $y \in Y$, let $X_y$ be the chain of 	nodes $x$ which are witnessed by $y$. By the lexicographic assumption, if $y,y' \in Y$ are such that $y$ is lexicographically before $y'$, then no element from $X_y$ has an ancestor in $X_{y'}$. Consequently, if we define $\pi_y$ to be the guide that connects all $X_y$ to $y$, then $\Pi=\set{\pi_y}_{y \in Y}$ is a one-color guidance system for all nodes in $G \cap X_3$.
\end{proof}

The guides of $\Pi$ we call \emph{starting guides} as usual.

We now need to combine $\Pi$ and $\Pi_\Gg$.
If we simply combine $\Pi_\Gg$ and $\Pi$, we might end up with a starting guide going through and external edge of $G$  that is already traversed by two transit guides. To avoid this problem, we need to do an optimisation
relying on a simple observation formulated in the claim below. 

A descendant guide $\pi$ is called \emph{live} in a subfactor $G$ if
$\pi$ passes through $G$, and its target is \emph{not} in a child subfactor 
of $G$ (i.e., the target is in a proper descendant of some child of $G$,
or outside $F$). The idea is that the target
of $\pi$ satisfies the assumption $m_2 \ge 2$ from the 'point of view' of 
nodes in $G$.  Note that guides live in $G$ may be either transit or starting in $G$.

\begin{claim}
  Suppose that two consistent descendant guides $\pi$ and $\pi'$ are live in 
  a subfactor $G$ and exit $G$ through the same edge.
  Suppose also that at least one of them is starting in $G$. 
  Then all source nodes of one of $\pi$, $\pi'$ can be moved to the other.
\end{claim}
\begin{proof}
  Let $(y_1,y_2)$ be the witness pair corresponding to $\pi$ and let
  $(y'_1,y'_2)$ be the witness pair corresponding to $\pi'$ --  
  by consistency, not only the second witnesses $y_2$ and $y'_2$ 
  are determined by $\pi$ and $\pi'$, but the whole witness pairs.
  Note that the source nodes of a descendant guide in $G$ are all situated on one path
  from the root of $G$ to one of the leaves.
  If both $\pi$ and $\pi'$ are starting in $G$, assume wlog that $\pi$ has a
  source node that is an ancestor of all source nodes of $\pi'$; otherwise 
  one of the guides is starting in $G$, wlog assume it is $\pi'$.
   As only the case of $m_1, m_2 \geq 2$
   is considered, and the values $s_1, \ldots, s_6$ are fixed,
   due to equation~\eqref{eq:forwApplied}(i)
   the pair $(y_1, y_2)$ is a witness for all source nodes
   of $\pi'$ as well. 
   Thus, these nodes may be guided to $y_2$ instead of $y'_2$.
\end{proof}

We use the
term \emph{live transit guide} for the transit guides that are live in
$G$, and \emph{dead transit guide} for the other transit guides
(those that have their target in a child subfactor of $G$).

Consider an edge $e$ that connects $G$ with a child subfactor $H$.
Suppose first that $e$ is traversed by a starting guide and a live  transit guide. Using the claim we merge the starting guide with the transit one.
Therefore, we end up satisfying the invariant property: 
$e$ is passed by at most one live guide, and possibly by one dead transit guide.

We have thus completed the proof of the Main Lemma and thus also
of Theorem~\ref{thm:simplified-bounded-width}.

\removed{
Suppose first that there are two dead transit guides that traverse
this edge. By applying the claim to the parent subfactor of $G$, 
we can move all the source nodes of
one of these guides to the other, and delete one guide. Therefore, we
may assume without loss of generality that there is at most one
dead transit guide  that passes through $e$.

Now suppose that some live transit guide passes through $e$. If any
starting guide passes through $e$, then we can again use the claim,
applied to $G$, to
move all the source nodes of the starting guide to the live transit
guide, and then delete the starting guide. Therefore, we end up
satisfying the invariant property: $e$ is passed either exclusively by
live transit guides (at most $n$), or starting guides (at most $n$)
plus at most one dead transit guide.
}


%% file: applications.tex
\section{Applications}
In this section, we present two applications of our results. The first application is a class of XML documents for which emptiness of XPath is decidable. The second application is a proof that two-variable first-order logic is not captured by XPath, in the presence of two attribute values per node. 

\sredniparagraph{Satisfiability of XPath}
As we said in the introduction, our study on class automata is a first step in a  search for structural restrictions on data words and data trees which make XPath satisfiability decidable. 

One idea for a  structural restriction would be  a variant of bounded clique width, or tree width. 
Maybe bounded  clique or tree width are interesting restrictions, but they are not relevant in the study of class automata. This is because bounded clique width or tree width, when defined in the natural way for data trees,  guarantees  decidable satisfiability  for a  logic far more powerful than class automata:  MSO with navigation and equal data value predicates.

Here we provide a basic example of a  restriction on inputs that works for class automata but not for MSO. A data tree is called  \emph{bipartite} (bipartite refers to the data) if  its nodes can be split into two connected (by the child relation) sets $X,Y$  such that every class has at most one node in $X$ and at most one node in $Y$. 

Satisfiability of MSO, or even FO, with navigation and data equality predicates
is undecidable even for bipartite data words. For instance, a solution to the 
Post Correspondence Problem can be encoded in a bipartite data word using a FO formula. 

This coding, however, cannot be captured by class automata, which is implied by the following theorem.

\begin{theorem}\label{thm:emptiness-bipartite}
	On bipartite data trees, emptiness is decidable for class automata, and therefore also
	for  XPath. 
\end{theorem}
\begin{proof}
The key insight is that data trees which use each data value at most twice can be described using semilinear sets.
To avoid notational complications, assume that every data value appears \emph{exactly} twice in a bipartite tree.
	
	Consider a class automaton  $\Aa$ over input alphabet $\Sigma$. 
Suppose that the work alphabet is $\Gamma$, and the  transducer is $f$. 
Let  the class condition be a language over alphabet $\Gamma \times \set{0,1}$, recognized	
	by a deterministic bottom-up tree automaton $\Cc$, with states $Q$.
For technical convenience, assume that a run of $\Cc$ labels tree edges, instead of tree nodes, with states.
Thus all leaf dummy edges of an input tree $u$ over $\Gamma \times \set{0,1}$ are labeled with the initial 
state of $\Cc$, and the labeling of all other edges is uniquely determined by $u$.
Let $\Cc(u)$ denote the state that labels the root dummy edge of $u$.
A tree $u$ is accepted if $\Cc(u)$ is an accepting state.

By assumption, the nodes of a bipartite data tree $(t, \sim)$ are partitioned into two connected subsets,
and thus there is a single edge that splits the two subsets; call this edge the \emph{border edge}.
Such a tree $t$ with a distinguish edge may be modeled as $t \otimes \set{z,z'}$, where $z$ and $z'$
are the two nodes connected by the border edge.  
Further, $t$ may be split into two smaller trees, according to the partition:
the border edge is a dummy root edge for one of the trees, and a leaf dummy edge for the other one.
Call these two induced trees \emph{lower} and \emph{upper} tree, respectively.

For the lower tree of $t$, call it $t_l$, and a subset $X$ of nodes of $t_l$, 
it makes sense to write $\Cc(t_l \otimes X)$. For the upper one, call it $t_u$,
we will need a slightly different notation. Assume that the automaton $\Cc$ reads $t_u \otimes X$, for some $X$,
starting in the initial state in all dummy leaf edges of $t_u$ except for the border edge,
where the automaton starts in some chosen state $q$. If $\Cc$ accepts the tree $t_u \otimes X$ under this assumption,
we write $q \in \Cc^{-1}(t_u \otimes X)$.
In particular, observe that the automaton $\Cc$ accepts $t \otimes X$ if and only if
$\Cc(t_l \otimes X_l) \in \Cc^{-1}(t_u \otimes X_u)$, where $X = X_l \cup X_u$ is the induced partition of $X$.

Given a tree $t$ over $\Gamma$ with a single distinguished edge, say $t \otimes \set{z, z'}$,
let $\pi(t \otimes \set{z, z'})$ be the set of all trees $s$ over $Q \times \set{l,u}$ that satisfy the following conditions:
\begin{iteMize}{$\bullet$}
\item the set of nodes of $s$ is the same as the set of nodes of $t$,
\item if a node $x$ of $s$ is in $t_l$ then it is labeled by $(\Cc(t_l \otimes \set{x}), l)$,
\item if a node $x$ of $s$ is in $t_u$ then it is labeled by $(q, u)$, for some $q \in \Cc^{-1}(t_u \otimes \set{x})$.
\end{iteMize}
Using $\pi$, we define the relation $\sigma$ between trees over $\Gamma$ and trees over $Q \times \set{l,u}$
as follows: $(t, s) \in \sigma$ if $s \in \pi(t \otimes \set{z,z'})$ for some edge $(z,z')$ in $t$. 
\begin{claim}
The relation $\sigma$ is computable by a nondeterministic letter-to-letter transducer.
\end{claim}
Let $T_\Sigma$ denote the set of all trees over $\Sigma$.
As nondeterministic transducers preserve regular languages, we have:
\begin{claim} \label{c:reg}
Both $K = f(T_\Sigma)$ and $\sigma(K) = \set{s : (t,s) \in \sigma, \ t \in K}$ are effectively computable regular languages over $\Gamma$ and
$Q \times \set{l,u}$, respectively.
\end{claim}
\begin{claim} \label{c:iff}
A class automaton $\Aa$ accepts some bipartite data tree if and only if $\sigma(K)$ contains a tree
that satisfies, for every $q \in Q$, the following condition:
the number of nodes labeled by $(q,l)$ is the same as the number of nodes labeled by $(q,u)$.
\end{claim}

With the last two observations decidability follows immediately.
To decide emptiness of a given class automaton $\Aa$, 
compute the Parikh image of the language $\sigma(K)$, an effectively semi-linear set by Claim~\ref{c:reg}, 
intersect this set with the semi-linear condition of Claim~\ref{c:iff},
and ckeck for emptiness of the resulting semi-linear set.
\end{proof}

\sredniparagraph{Multiple attributes}
\label{sec:multiple}
Heretofore, we have studied data trees, which model XML documents where each node has one data value. In this section, and this one only, we consider the situation where  each node $x$ has $n$ data values. Formally, an $n$-data tree consists of a tree $t$ over the finite alphabet and functions $d_1,\ldots,d_n$ which map the tree nodes to data values. How does XPath deal with multiple data values? Instead of $y_1 \sim y_2$ and $y_1 \not \sim y_2$, we can use any formula of the form
\begin{eqnarray*}
	d_i(y_1) =  d_j(y_2) \qquad \mbox{ where }i,j \in \set{1,\ldots,n}
\end{eqnarray*}
or its negation (for inequality). For $n > 1$ we need more information than just the partitions of nodes into classes
of~$\sim_i$, for each $i \in \set{1,\ldots,n}$. An example is the property ``every node has the same data value on attributes $1$ and $2$''.

How do we extend class automata to read $n$-data trees? For one data value, the class condition is a language over the alphabet $\Gamma \times \set{0,1}$. For $n$ data values, the class condition is a language over the alphabet $\Gamma \times \set{0,1}^n$. An $n$-data tree  $(t,d_1,\ldots,d_n)$ is accepted if there is an output $s$ of the transducer on $t$ such that for every data value $d$, the tree
\begin{eqnarray*}
	s \otimes d_1^{-1}(d) \otimes \cdots \otimes d_n^{-1}(d)
\end{eqnarray*}
is accepted by the class condition. By the same technique as in the proof of Theorem~\ref{thm:main}, we can  prove that the automata capture XPath. 

A consequence  is that for $n  \ge 2$,  XPath does not capture two-variable first-order logic (unlike the case of $n=1$). This was an open question.
\begin{theorem}\label{cor:two-variable}
	The following (two-variable) property
	\begin{eqnarray*}
		\psi = \forall x \forall y \quad d_1(x)=d_1(y) \iff d_2(x)=d_2(y)
	\end{eqnarray*}
	cannot be defined by a boolean query of XPath.
\end{theorem}
\begin{proof}
Towards a contradiction, suppose that $\psi$ is recognized by any class automaton in the generalized version for 
2 data values.  The document that will confuse the automaton will be a word.
	Consider the document (over a one letter alphabet) with positions $1,\ldots,2n$, where the data values are defined
	\begin{eqnarray*}
		d_1(i)=d_1(n+i)=i &\mbox{ for $i \in \set{1,\ldots,n}$}\\
		d_2(i)=d_2(n+i)=-i& \mbox{ for $i \in \set{1,\ldots,n}$}
	\end{eqnarray*}
	Since the above document satisfies $\psi$, it should also be accepted by the automaton. Let the work alphabet of the automaton be $\Gamma$, and let $a_1 \cdots a_{2n} \in \Gamma^{*}$ be the word produced by the automaton in the accepting run.  For a data value $d$, we use the term \emph{class string} of $d$ for the word 
	\begin{eqnarray*}
		a_1 \cdots a_{2n} \otimes d_1^{-1}(d) \otimes d_2^{-1}(d).
	\end{eqnarray*}
	By definition of the way class automata accept documents, each class string should belong to the class condition. 
	Consider a number $i \in \set{1,\ldots,n}$. The class string of $i$ is
\begin{eqnarray*}
u_i =&	a_1 \cdots a_{2n} \otimes \set{i,n+i} \otimes \emptyset
\end{eqnarray*}
Suppose that
\begin{eqnarray*}
	\alpha : (\Gamma \times \set{0,1} \times \set{0,1})^* \to S
\end{eqnarray*} is a morphism recognizing the class condition.  
If $n$ is greater than $|S|^2$, then we can find two data values $i < j \in \set{1,\ldots,n}$ such that
\begin{eqnarray*}
	\alpha(u_i | \set{1,\ldots,n}) &=&\alpha(u_j | \set{1,\ldots,n})\\
	\alpha(u_i | \set{n+1,\ldots,2n}) &=& \alpha(u_j | \set{n+1,\ldots,2n}) .
\end{eqnarray*}
Consider a new document  obtained from the previous one by swapping the first, but not second, data value in the positions $i$ and $j$. This new document violates the property $\psi$. To get the contradiction,  we will construct an accepting run of the class automaton for this new document. The output of the transducer is the same $a_1 \cdots a_{2n}$. The class strings are the same for data values other than $i$ and $j$, so they are also accepted. For the class strings of $i$ and $j$, the images under $\alpha$ are the same by assumption on $i$ and $j$, and hence they are also accepted.
\end{proof}